\documentclass[11pt]{article}

\usepackage{amsmath}
\usepackage{amsthm}
\usepackage{amssymb}
\usepackage{graphicx}
\usepackage{bm}
\usepackage[margin=1.1in]{geometry}
\usepackage{hyperref}

\newcommand{\rmd}{\mathrm{d}}
\newcommand{\Tr}{\mathrm{Tr}}

\theoremstyle{plain}
\newtheorem{thm}{Theorem}[section]
\newtheorem{prop}[thm]{Proposition}
\newtheorem{lem}[thm]{Lemma}
\theoremstyle{remark}
\newtheorem{rem}[thm]{Remark}

\title{Lattice Green's function of the hyperkagome lattice: modular uniformization at level~30
from an orthogonal differential Galois group}

\author{%
Bryan Nasr\thanks{Independent researcher.\ \texttt{bryannasr4@gmail.com}}
\and
Jean-Marie Maillard\thanks{LPTMC, UMR 7600 CNRS, Sorbonne Universit\'e, Tour 23, 5\`eme
\'etage, case 121, 4 Place Jussieu, 75252 Paris Cedex 05, France.\
\texttt{jean-marie.maillard@sorbonne-universite.fr}}%
}

\date{}

\begin{document}

\maketitle

\begin{abstract}
The lattice Green's functions of the cubic lattices are classical: each is a symmetric square of a
second-order operator, with a closed form in complete elliptic integrals. The hyperkagome lattice,
realized by the iridium sublattice of the spin-liquid candidate
Na\textsubscript{4}Ir\textsubscript{3}O\textsubscript{8}, has resisted such
a treatment: Varma and Monien reduced its density of states to a threefold integral they found no way
to solve exactly. From exact lattice moments we obtain an irreducible third-order linear differential
(Picard--Fuchs) operator annihilating the Green's function. It is not a literal symmetric square, yet its symmetric
square carries a rational solution: the differential Galois group is orthogonal, and an order-two
differential intertwiner makes the operator projectively a symmetric square. Our main result is that the
underlying second-order operator is the uniformizing equation of the genus-zero modular curve of
$\Gamma_0(30)^+$, the level-thirty group generated by $\Gamma_0(30)$ and all its Atkin--Lehner
involutions, an explicit eta quotient parametrizing the natural variable. That uniformizing
equation is itself not new---it is a row of the Lian--Yau table of genus-zero groups---so what is
new here is the identification of a lattice Green's function with it, together with a proof: the
Schwarzian identity is established by an a priori pole-degree bound rather than checked
numerically, which also proves the tabulated entry. The hyperkagome Green's function is therefore
modular at level thirty, supplying the closed form Varma and Monien sought. Its weight-two period
is obtained explicitly: a depth-one quasimodular form which we prove is not a modular form times
an algebraic function at any weight.
\end{abstract}



\section{Introduction}
The lattice Green's function (LGF)
\begin{equation}
G(z)=\int_{\mathrm{BZ}}\frac{\rmd^3k}{(2\pi)^3}\,\frac{1}{z-\varepsilon(\bm k)}
=\int\frac{\rho(E)\,\rmd E}{z-E},
\label{eq:lgf}
\end{equation}
the resolvent of a tight-binding band $\varepsilon(\bm k)$ averaged over the Brillouin zone~\cite{AshcroftMermin}, is a central object
of lattice statistics: it controls random-walk return probabilities, defect and impurity states, and the local
density of states $\rho(E)=-\pi^{-1}\,\mathrm{Im}\,G(E+i0^+)$~\cite{Economou,Guttmann2010}. For the
three-dimensional cubic lattices the LGF is a celebrated closed form. Watson's triple
integral for the simple-cubic (SC) lattice~\cite{Watson1939}, and the body- and face-centered analogues, are all
expressible through products of complete elliptic integrals $K$~\cite{Joyce1973,GlasserZucker1977,Joyce1998}; structurally,
each of these order-three linear differential operators is the \emph{symmetric square} of a second-order
(elliptic) one~\cite{BoukraaHassaniMaillard2014}, which is precisely why an elliptic closed form exists.

Not every lattice is so fortunate. As the local geometry becomes more intricate---decorated, frustrated, or
multi-band lattices---the LGF can fail to reduce to elliptic integrals and instead becomes a period of a
higher-order or irreducible linear differential operator (the diamond lattice already signals this
trend~\cite{Guttmann2010,Joyce1998}). The natural three-dimensional frustrated network is the
\emph{hyperkagome} lattice: a chiral arrangement of corner-sharing triangles, obtained by removing one quarter
of the sites of the pyrochlore lattice in an ordered (space group $P4_132$) fashion. It is realized by the Ir$^{4+}$
moments of Na\textsubscript{4}Ir\textsubscript{3}O\textsubscript{8}, a leading three-dimensional quantum-spin-liquid candidate~\cite{Okamoto2007,
Hopkinson2007,Lawler2008}, which has made its single-particle spectral functions of independent interest.

Varma and Monien~\cite{VarmaMonien2013} computed the hyperkagome density of states and reduced it to a single
threefold integral of a rational function of $\cos k_x,\cos k_y,\cos k_z$, remarking that ``it seems highly
suggestive that a closed form expression can be obtained if a particular complicated yet highly symmetrical
integral can be solved,'' but that ``we have currently found no way to exactly solve'' it. That integral has
remained open.

We settle its analytic structure and prove its modularity. The hyperkagome LGF is governed by an \emph{irreducible third-order}
linear differential (Picard--Fuchs) operator $M$ that is not a \emph{literal} symmetric square, yet whose symmetric square carries a
rational solution: the differential Galois group is \emph{orthogonal}, $G_M=\mathrm{O}(3,\mathbb{C})$ with
$G_M^\circ=\mathrm{SO}(3,\mathbb{C})\cong\mathrm{PSL}(2,\mathbb{C})$, so $M$ \emph{is} projectively equivalent to the
symmetric square of a second-order operator. The obstruction Varma and Monien met is therefore not a genuine
third-order (``$\mathrm{SL}_3$-generic'') object with no elliptic reduction; it is a \emph{hidden orthogonal /
modular structure}, realized through a nontrivial operator equivalence rather than the function-multiplier and
pullback equivalences one usually tests. We obtain the operator explicitly, certify this structure by exact
computation, and thereby explain \emph{why} a naive symmetric-square reduction was never found. We then prove the
closed-form structure outright: the second-order operator underlying $M$ is the uniformizing operator of the
genus-zero modular curve $X(\Gamma_0(30)^+)$, the natural variable $t$ generates its function field and is given
by an explicit level-30 eta-quotient parametrization, and the hyperkagome LGF is modular at level
$30=2\cdot3\cdot5$---to our knowledge the first
three-dimensional LGF known to require such an equivalence, and the first lattice Green's function realized at a
modular level with three distinct prime factors. The uniformizing operator itself is not new: it is tabulated,
as the row $\Gamma_0(30)^+$ of a list of genus-zero groups, in~\cite{LianYau1995}. What is new here is the
identification of a lattice Green's function with that operator, and the proof: Theorem~\ref{thm:modular}
establishes the Schwarzian identity by an a priori pole-degree bound rather than by a series match, and thereby
also proves the tabulated entry.

\section{Hyperkagome lattice and its Green's function}
\label{sec:lattice}
The hyperkagome lattice has twelve sites per cubic cell, each with coordination number four; its nearest-neighbor
tight-binding Bloch Hamiltonian $H(\bm k)$ is a $12\times12$ Hermitian matrix. We fix the hopping amplitude to
unity (reserving $t$ for the spectral variable introduced below) and the energy origin at the band center.
Diagonalizing $H(\bm k)$ over the Brillouin zone yields the spectrum shown in
Figure~\ref{fig:dos}. Three exact features organize everything that follows.

\emph{(i) A flat band, and the line-graph structure.} The hyperkagome lattice is the line graph of a three-regular
``premedial'' net $\Gamma$ (eight vertices per cell, twelve edges): its sites are the edges of $\Gamma$ and its
triangles are the vertices of $\Gamma$. For the line graph $L(\Gamma)$ of a $d$-regular graph one has
$A(L(\Gamma))=B^{\!\top}B-2I$ with $B$ the incidence matrix, so the spectrum is $\{\mu_i+d-2\}$ together with the
eigenvalue $-2$ of multiplicity $(\text{edges})-(\text{vertices})$. Here $d=3$, giving four flat bands at $E=-2$
(spectral weight exactly $1/3$, i.e.\ $\tfrac{12-8}{12}$); this contributes a pole $\tfrac13(z+2)^{-1}$ to $G(z)$
and is the flat-band pole written $1/(t_{\rm VM}{+}1)$ in \cite{VarmaMonien2013}, whose energy variable is
$t_{\rm VM}=E/2$ (the correspondence, including this pole and its weight, is reproduced in their variable by
\texttt{numerics/vm\_crosscheck.py} in the repository). The eight dispersive bands are $E=\mu_i+1$ and fill $[-2,4]$, since the adjacency
eigenvalues $\mu_i$ of a three-regular graph lie in $[-3,3]$.

\emph{(ii) An exact reflection symmetry, from bipartiteness.} A graph has an adjacency spectrum symmetric about
$0$ if and only if it is bipartite; we verify that the premedial net $\Gamma$ is bipartite (it admits a
translation-consistent two-coloring, checked on finite blocks up to $5^3$ cells). Through the line-graph shift $E=\mu+1$ this makes the dispersive density of states \emph{exactly}
symmetric about $E=1$; independently, all odd central moments
$\int\rho_{\rm disp}(E)(E-1)^{2m+1}\rmd E$ vanish in exact arithmetic for every odd order up to
$229$, the highest the tabulated moments support. Hence the dispersive part of
$G$ is odd in $\zeta\equiv z-1$, and is a function of $t\equiv\zeta^{-2}$.

\emph{(iii) Integer moments.} The per-site spectral moments $m_n=\frac{1}{12}\Tr H^n$ are exact integers,
$m_0,\dots,m_{10}=1,0,4,4,28,60,260,756,2828,9292,33384$; we computed them to $m_{230}$ by closed-walk enumeration.
Among the ordered quarter-depletions of the pyrochlore lattice, the chiral hyperkagome enantiomorphs are the ones
reproducing this sequence (hexagon-containing depletions already differ at $m_6=264\neq260$); this selection is
checked separately and is not needed for what follows.

\emph{(iv) From the raw moments to $t$: a pullback of degree two.} Features \emph{(i)} and \emph{(ii)} are what
carry the raw moments into the variable used from Section~\ref{sec:pf} onward, and the passage is worth writing
out, because it is a pullback of degree two and not a change of variable. Set $x=1/z$ and let
$S(x)=\sum_{n\ge0}m_nx^n=z\,G(z)$ be the generating function of the moments themselves. Then
\begin{equation}
S(x)=\frac{1/3}{1+2x}+\frac{1}{1-x}\,\Phi\!\left(\frac{x^{2}}{(1-x)^{2}}\right),
\qquad \Phi(t)=\sum_{m\ge0}\nu_m t^m,
\label{eq:momentbridge}
\end{equation}
an identity verified coefficient by coefficient in exact rational arithmetic for $x^{0}$ through $x^{230}$, the
full extent of the tabulated moments (\texttt{numerics/verify\_moment\_bridge.py}). The first term is the
flat-band pole of \emph{(i)}; the factor $(1-x)^{-1}=z/\zeta$ is the gauge implied by $\zeta G_{\rm disp}$ being
even in $\zeta$; and
\begin{equation}
t=\frac{x^{2}}{(1-x)^{2}}=\frac{1}{(z-1)^{2}}
\label{eq:tofx}
\end{equation}
is the degree-two invariant of the involution $x\mapsto x/(2x-1)$, which is the reflection $E\mapsto2-E$ of
\emph{(ii)} read in $x$. The passage to $t$ is therefore a \emph{quadratic pullback}, not a substitution, and
every singular point of the operator $M$ of Section~\ref{sec:pf} accordingly has two preimages in the energy
variable, one on either side of $E=1$:
{\small\setlength{\arraycolsep}{3pt}%
\begin{equation}
\begin{array}{c|cccccc}
\text{locus in }t & 9t-1 & 5t-1 & 4t-1 & t-1 & t=0 & t=\infty\\\hline
\text{preimage in }x & (4x-1)(2x+1) & 4x^{2}+2x-1 & (3x-1)(x+1) & 2x-1,\ x=\infty & x=0 & x=1\\
\text{energies} & E=4,\,-2 & E=1\pm\sqrt5 & E=3,\,-1 & E=2,\,0 & E=\infty & E=1
\end{array}
\label{eq:branches}
\end{equation}}%
the last two being double points, and the degree-seven locus pulling back to
$-(1-x)^{14}p_7\!\left(x^{2}/(1-x)^{2}\right)$. The pairing is not the one the constants suggest: it is $9t-1$,
and not $4t-1$, that carries the factor $4x-1$. One consequence is worth recording for a reader who starts, as
is natural, from the integer moments themselves. A minimal annihilator guessed directly from the $m_n$ in the
variable $x$ is not the $L_4$ of Section~\ref{sec:pf} but an operator of order \emph{five}: the least common
left multiple of an order-one operator carrying the flat-band pole with the order-four pullback of $L_4$, which
in turn factors as $\mathcal{N}_3\cdot\left(\rmd/\rmd x+1/(x-1)\right)$, where $\mathcal{N}_3$ is $M$ pulled
back along \eqref{eq:tofx} and regauged. That order-one right factor is the trivial factor $\rmd/\rmd t$ of
$L_4$ wearing the gauge of \eqref{eq:momentbridge}, and its solution $1/(1-x)$ is the constant solution in that
gauge---not an independent spectral feature at $E=1$. Note finally that pulling back a product is not the same
as multiplying the pullbacks: the Jacobian of \eqref{eq:tofx} intervenes, exactly as it does for $Q_V$ under the
Schwarzian in Section~\ref{sec:modular}.

\emph{(v) The energy variable is itself modular.} Section~\ref{sec:modular} will show that $t$ generates the
genus-zero function field of $X(\Gamma_0(30)^+)$, with the cusp at $t=0$ and $t=\infty$ one of the five order-two
points. Since \eqref{eq:tofx} is branched over exactly those two points, $h=x/(1-x)=1/(E-1)=\sqrt{t}$ generates
the function field of $X(G)$ for a subgroup $G\subset\Gamma_0(30)^+$ of index two, of signature $(0;2^{8};1)$
(genus zero, eight order-two points, one cusp, covolume $6\pi$), and $t=h^{2}$. Which subgroup is decided by the
$q$-expansion: $t=q+O(q^{2})$ exactly, so $h=q^{1/2}(1-2q+4q^{2}-9q^{3}+\cdots)$ and $h(\tau+1)=-h(\tau)$. Hence
$G$ does not contain $\Gamma_0(30)$ and is none of the seven groups $\Gamma_0(30)+W$, and the reflection
$E\mapsto2-E$ is \emph{not} an Atkin--Lehner involution: it is the translation $\tau\mapsto\tau+1$ acting on
$X(G)$. Explicitly, $G\cap\Gamma_0(30)=\{\gamma:\ b+c/30\equiv0\pmod 2\}$, a congruence subgroup of level $60$
(genus $9$, eight cusps, every cusp width of $\Gamma_0(30)$ doubled), and
$h=\kappa\,F/H$ with $\kappa=(\eta(\tau)/\eta(2\tau))^{12}$ and explicit $F,H\in M_4(\Gamma_0(30))$, an identity
certified in $M_{32}(\Gamma_0(30))$ past its valence bound $192$
(\texttt{numerics/certify\_energy\_hauptmodul.py}). The level-$30$ statement of Section~\ref{sec:modular} is thus
a statement about $t$; the statement about $E$ itself is one at level $60$.

\begin{figure}[t]
\centering
\includegraphics[width=0.8\textwidth]{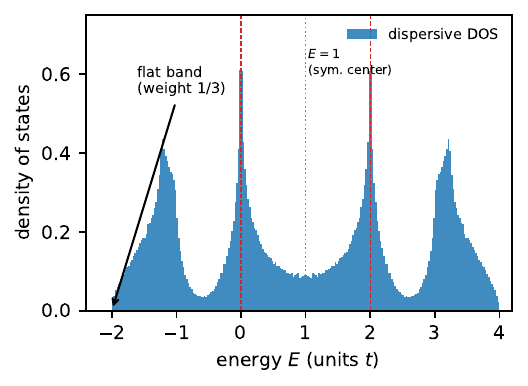}
\caption{Density of states of the hyperkagome tight-binding model at unit hopping. Four flat bands at $E=-2$ carry weight
$1/3$; the eight dispersive bands fill $[-2,4]$ and are \emph{exactly} symmetric about $E=1$, with principal van
Hove singularities at $E=0,2$ (dashed) and weaker interior features at $E=-1,3$ and $E=1\pm\sqrt5$ (see text).}
\label{fig:dos}
\end{figure}

\section{The Picard--Fuchs operator}
\label{sec:pf}
Because $G(z)$ is the diagonal of a rational function over the torus, it is holonomic (D-finite): it satisfies a
linear differential equation with polynomial coefficients~\cite{Lipshitz1988}. We determine that equation from
the moments. Throughout, \emph{operator} abbreviates \emph{linear differential operator} in $\rmd/\rmd t$ with
coefficients in $\mathbb{Q}(t)$, and all products, factorizations, adjoints and symmetric squares are taken in
the ring $\mathbb{Q}(t)\langle\rmd/\rmd t\rangle$ of such operators. Removing the flat band, $d_n=m_n-\tfrac13(-2)^n$, and forming the even central moments
$\nu_m=\int\rho_{\rm disp}(E)\,(E-1)^{2m}\rmd E$, the generating function $\Phi(t)=\sum_m\nu_m t^m$ (with
$t=\zeta^{-2}$) is the essential object. Guessing a minimal annihilator from the exact $\nu$ sequence by the
method of differential approximants~\cite{Guttmann2010,orealgebra} returns a linear differential operator $L_4$
of order four and degree fifteen.

This $L_4$ is reducible in a trivial way: it admits the constant solution, so $L_4=M\,\rmd/\rmd t$ exactly, with
$M$ of order \emph{three} and degree fifteen,
\begin{equation}
M=\sum_{i=0}^{3}c_i(t)\,\Big(\frac{\rmd}{\rmd t}\Big)^{i},\quad \deg c_i\le 15,
\label{eq:M}
\end{equation}
the minimal annihilator of $\Phi'$---minimal in the certified, guess-and-verify sense
made precise under \emph{Certification} below. That factorization is moreover the \emph{only} one:
$L_4$ is not a least common left multiple $\mathrm{LCLM}(N,\rmd/\rmd t)$ of $\rmd/\rmd t$ with a
second operator of order three, so the constant solution is not a direct summand and $M$---rather
than some other order-three operator---is the intrinsic object here
(Proposition~\ref{prop:unique} below). Its leading coefficient factorizes as
\begin{equation}
c_3(t)\propto t^2\,(t-\tfrac19)(t-\tfrac15)(t-\tfrac14)(t-1)^3\,p_7(t),
\label{eq:lc}
\end{equation}
with $p_7$ an irreducible degree-seven polynomial. The finite singular points are images of the band structure
under $t=(E-1)^{-2}$: the band edge $E=4$ (equivalently $E=-2$) maps to $t=\tfrac19$, the principal van Hove
energies $E=0,2$ to $t=1$, and the interior van Hove features $E=-1,3$ and $E=1\pm\sqrt5$ to $t=\tfrac14$ and
$t=\tfrac15$; the half-integer exponent in \eqref{eq:riemann} confirms a square-root singularity at each. The
remaining singular points ($t=0$, $t=\infty$, and the degree-seven locus $p_7$) all appear in the complete Riemann
scheme \eqref{eq:riemann} below.

\emph{Certification.} The operator $M$ was reconstructed exactly over $\mathbb{Q}$ and shown to annihilate the
series $\Phi'(t)$ \emph{exactly over $\mathbb{Q}$} through order $t^{111}$---$112$ independent relations against
the $57$ free coefficients of $M$ (one of the $58$ stored coefficients fixes the overall scale), an
overdetermination margin of $55$ with zero residuals. It was
furthermore re-derived from a second, independent implementation by the authors---an independent
quarter-depletion construction, an independent closed-walk moment enumeration, and an independent
guess---and the moment sequence and all four coefficient polynomials of $M$ came out bit-identical.
As a further independent check, the
integer moments were reproduced by direct diagonalization of the $12\times12$ Bloch Hamiltonian on a
Brillouin-zone grid---a method disjoint from the closed-walk enumeration. This is the standard
``conjectured-then-verified'' certification of the lattice-statistics literature~\cite{Guttmann2010}. An
unconditional derivation by creative telescoping~\cite{Koutschan2013,Bostan2013} from the rational integrand of
the resolvent---which would upgrade this verification to a formal proof---is a natural next step, well suited to
a dedicated holonomic-functions implementation.

\section{Riemann scheme and differential-Galois structure}
The analytic nature of the LGF is decided by the local structure of $M$. Computing the Frobenius indicial
exponents at each regular singular point gives the Riemann scheme
{\small\setlength{\arraycolsep}{3pt}%
\begin{equation}
\begin{array}{c|ccccccc}
 & t=0 & t=\tfrac19 & t=\tfrac15 & t=\tfrac14 & t=1 & p_7 & t=\infty \\\hline
\text{exp.} & \{-1,-1,0\} & \{-\tfrac12,0,1\} & \{-\tfrac12,0,1\} & \{-\tfrac12,0,1\}
 & \{-1,-\tfrac12,\tfrac12\} & \{0,1,3\} & \{\tfrac32,2,3\}
\end{array}
\label{eq:riemann}
\end{equation}}
Here $p_7$ denotes the irreducible degree-seven locus, whose seven Galois-conjugate roots all carry the same
integer exponents $\{0,1,3\}$, computed exactly by reduction modulo $p_7$. These are \emph{apparent}
singularities: integer exponents alone do not establish this, but all three local solutions at each root are
log-free---indeed analytic. Log-freeness is verified directly through the Frobenius obstruction at each root of
$p_7$, computed modulo several independent primes; independently of that check, analyticity follows
unconditionally from the intertwiner of Section~\ref{sec:y0closed}: the exact operator identity of step [IV] of
Theorem~\ref{thm:y0}, together with the injectivity computation \eqref{eq:Einj}, carries a full solution basis of
the operator $\widetilde N$---which is regular on the $p_7$ locus, as are the twist unit $v$ and the first-order
operator $T$ of \eqref{eq:y0rhos}---onto a basis of solutions of $M$, so every local solution of $M$ at each root
of $p_7$ is analytic. We use the standard
Riemann-scheme convention at infinity, in which an exponent $\rho$ at $t=\infty$ corresponds to a local solution
$y\sim t^{-\rho}$; the three solutions there thus behave as $t^{-3/2}$, $t^{-2}$ and $t^{-3}$. As a check on
\eqref{eq:riemann}, the operator has $s=13$ singular points---the five rational points above, the seven roots of
$p_7$, and $t=\infty$; the relation counts the apparent singularities on $p_7$ along with the rest---and the
exponents sum to
\begin{equation}
-2 + 3\cdot\tfrac12 - 1 + 7\cdot 4 + \tfrac{13}{2} \;=\; 33 \;=\; \tfrac{n(n-1)}{2}\,(s-2),\qquad n=3,
\end{equation}
so the Fuchs relation is satisfied exactly.

The repeated exponent $-1$ at $t=0$ is the source of the logarithmic solution used
below. The half-integer exponent $-\tfrac12$ at the band edge $t=\tfrac19$ is the square-root van Hove singularity
of a three-dimensional band edge, $\rho(E)\sim\sqrt{E_{\rm edge}-E}$.

We first record two local-to-global facts about $M$---it is irreducible, and it is not a \emph{literal} symmetric
square---and then determine its differential Galois group, which is the object that actually decides the closed-form
question.

\emph{$M$ is irreducible over $\overline{\mathbb{Q}}(t)$.} $M$ is Fuchsian---regular singular at every point,
including $t=\infty$ (exponents $\{\tfrac32,2,3\}$)---so any first-order right factor $\mathrm{d}/\mathrm{d}t-h$
has $h=\sum_s e_s/(t-s)$ with $e_s$ a local exponent at each singular point $s$ and \emph{no} polynomial part. At
the rational singular points the $e_s$ are the exponents of \eqref{eq:riemann}; at the irreducible degree-seven
locus $p_7$ the common (Galois-invariant) residue ranges over its exponents $\{0,1,3\}$. This candidate set is
finite, and each candidate is tested exactly over $\mathbb{Q}$ by the
divisibility criterion $c_3(h''+3hh'+h^3)+c_2(h'+h^2)+c_1h+c_0=0$. No first-order right factor of $M$ exists over
$\mathbb{Q}(t)$, and none of $\mathrm{adjoint}(M)$---which for a third-order operator excludes an order-two right
factor of $M$ as well---so $M$ has no right factor of order one or two over $\mathbb{Q}(t)$. We now descend to
$\overline{\mathbb{Q}}(t)$ in two steps. \emph{First}, $M$ has no order-one right factor there: a first-order right factor
$\mathrm{d}/\mathrm{d}t-h$ with $h\in\overline{\mathbb{Q}}(t)$ has a finite Galois orbit
$\{\mathrm{d}/\mathrm{d}t-h^\sigma\}$ of right factors of $M$, whose least common left multiple is a right factor
\emph{defined over $\mathbb{Q}(t)$} of order $k\in\{1,2,3\}$; $k=1$ gives $h\in\mathbb{Q}(t)$ and $k=2$ an
order-two right factor over $\mathbb{Q}$, both excluded above, while $k=3$ would make $M$ the LCLM of order-one
operators, hence completely reducible, forcing \emph{semisimple} local monodromy and no logarithm---contradicting
the genuine logarithm at $t=0$ (repeated indicial exponent $-1$, exactly one log-free local solution of three,
verified directly). \emph{Second}, $M$ has no order-two right factor over $\overline{\mathbb{Q}}(t)$: such a factor
is equivalent to an order-one right factor of $\mathrm{adjoint}(M)$ over $\overline{\mathbb{Q}}(t)$, excluded by the
identical descent applied to $\mathrm{adjoint}(M)$ (no order-one factor over $\mathbb{Q}$ by the adjoint search, no
order-two by the $M$ search, and $\mathrm{adjoint}(M)$ is not completely reducible, since complete reducibility is
self-dual and $M$ is not). Hence $M$ is irreducible over $\overline{\mathbb{Q}}(t)$.
(This exact enumeration replaces the floating-point candidate search of earlier drafts; the computation is carried
out by \texttt{numerics/certify\_factor.py} in the repository and is independently reproducible with a certified
factorizer such as Maple's \texttt{DFactor} or Magma.)

Irreducibility settles a second question, raised already in Section~\ref{sec:pf}: an order-four operator
with a constant solution admits the factorization $L_4=M\,\rmd/\rmd t$, but it may instead be the
\emph{direct sum} $\mathrm{LCLM}(N,\rmd/\rmd t)=N\oplus\rmd/\rmd t$, in which case $N$---not $M$---is the
operator to study, and every computation below would have to be redone for it. That does not happen here.

\begin{prop}\label{prop:unique}
$L_4=M\,\rmd/\rmd t$ is the unique factorization of $L_4$ into linear differential operators of positive
order. In particular $L_4\neq\mathrm{LCLM}(N,\rmd/\rmd t)$ for every order-three operator $N$: the constant
solution is not a direct summand of $\mathrm{Sol}(L_4)$.
\end{prop}

\begin{proof}
Factorizations of $L_4$ correspond to submodules of $\mathrm{Sol}(L_4)$, and $\rmd/\rmd t$ maps
$\mathrm{Sol}(L_4)$ onto $\mathrm{Sol}(M)$ with kernel the constants $\mathbb{C}$. Since $M$ is irreducible, a
one-dimensional submodule has image $0$ and therefore lies in $\mathbb{C}$, while a two-dimensional one would
have an image of dimension one or two---in either case a proper nonzero submodule of $\mathrm{Sol}(M)$. So besides $0$, $\mathbb{C}$ and
$\mathrm{Sol}(L_4)$ the only possible submodule is a three-dimensional $S$; such an $S$ meets $\mathbb{C}$
trivially (otherwise its image would again be two-dimensional), giving $\mathrm{Sol}(L_4)=S\oplus\mathbb{C}$,
i.e.\ exactly the LCLM case with $\mathrm{Sol}(N)=S$.

Suppose such an $S$ existed. Then $\rmd/\rmd t\colon S\to\mathrm{Sol}(M)$ is an isomorphism of differential
modules, and it restricts to one of the local solution spaces on a punctured disc around $t=1$, where we now
work. There the exponents of $M$ are $\{-1,-\tfrac12,\tfrac12\}$ by \eqref{eq:riemann}: the exponent $-1$ is a
simple root of the indicial equation and no exponent equals $-1+k$ for a positive integer $k$, so the Frobenius
solution belonging to it carries no logarithm. Being a Laurent series in integer powers it is single-valued,
and normalized it reads $y=(t-1)^{-1}\big(1+O(t-1)\big)$, of residue $1$.
Let $Y\in S$ be its preimage and let $\gamma$ be the local monodromy at $t=1$, so that $\gamma y=y$. As $\gamma$
lies in the monodromy group, hence in the differential Galois group, and $S$ is a submodule,
$\gamma Y-Y\in S$, while $(\gamma Y-Y)'=\gamma y-y=0$ puts it in $\mathbb{C}$; as $S\cap\mathbb{C}=0$ we get
$\gamma Y=Y$. But $Y'=y$ forces $Y=\log(t-1)+(\text{single-valued})$, so $\gamma Y=Y+2\pi i$---a contradiction.
\end{proof}

A second, global proof runs through the adjoint and is independent of the local one. Write $L_4=A\,N$ with $N$
a right factor of order three and $A$ of order one; since $\mathrm{adjoint}(A\,N)=\mathrm{adjoint}(N)\,
\mathrm{adjoint}(A)$, an order-three right factor of $L_4$ is the same thing as an order-one right factor of
$\mathrm{adjoint}(L_4)$, i.e.\ a hyperexponential solution of it. Now
$\mathrm{adjoint}(L_4)=-(\rmd/\rmd t)\,\mathrm{adjoint}(M)$, so such a solution obeys
$\mathrm{adjoint}(M)(y)=c$ for a constant $c$. Applying an operator with rational coefficients to a
hyperexponential $y$ returns $y$ times a rational function, $\mathrm{adjoint}(M)(y)=S\,y$ with
$S\in\overline{\mathbb{Q}}(t)$: hence $c=0$ would make $y$ a hyperexponential solution of the irreducible
$\mathrm{adjoint}(M)$, while $c\neq0$ gives $y=c/S$, a \emph{rational} solution of $\mathrm{adjoint}(L_4)$.
There is none---and it is enough to search over $\mathbb{Q}(t)$, the space of rational solutions being
Galois-stable, hence defined over $\mathbb{Q}$: the local exponents bound the denominator and the exponents at
infinity bound the degree, and the resulting finite linear system over $\mathbb{Q}$ has only the zero solution. Both proofs are certified in
\texttt{numerics/certify\_lclm.py} (\texttt{CERTIFICATE\_lclm.txt}), which first validates its Frobenius and
rational-solution primitives on operators of known structure---including a manufactured
$\mathrm{LCLM}(N,\rmd/\rmd t)$, on which the same search does find the rational solution, so the negative
result is not vacuous.

\emph{$M$ is not a literal symmetric square.} If $M$ were literally the symmetric square of a second-order operator
with local exponents $\{a,b\}$, its three exponents would be $\{2a,a+b,2b\}$, an arithmetic progression whose middle
term is the mean of the outer two. No triple in \eqref{eq:riemann} is of this form: at $t=\tfrac19$ the exponents
$\{-\tfrac12,0,1\}$ have mean of outer terms $\tfrac14\neq0$, and at $t=1$ the exponents $\{-1,-\tfrac12,\tfrac12\}$
have mean of outer terms $-\tfrac14\neq-\tfrac12$; a single such point already suffices. As a
positive control, the same analysis applied to the simple-cubic LGF returns the symmetric-square-compatible
exponents $\{0,\tfrac12,1\}=\{2a,a+b,2b\}$ with $\{a,b\}=\{0,\tfrac12\}$, correctly reflecting its closed form in
$K$. \emph{The scope of this exclusion is limited, and it is essential not to overstate it.} A gauge transformation
shifts the three exponents at each point by a common constant, and an algebraic pullback multiplies them by the
local ramification index; both preserve arithmetic progressions, so the exponent test also rules out every
\emph{function-multiplier or pullback} equivalent of a symmetric square. It does \emph{not}, however, rule out an
equivalence realized by a \emph{differential intertwiner}---a homomorphism of differential modules of order two,
which shifts the individual local exponents by \emph{different} integers and so destroys the arithmetic
progressions while preserving the Galois-theoretic content. This is exactly the loophole the hyperkagome operator
exploits: as we show in Section~\ref{sec:orthogonal}, $M$ \emph{is} equivalent to a symmetric square through such an
order-two intertwiner, even though it is not a literal symmetric square (nor a multiplier/pullback of one). Correcting an overstatement in
an earlier version of this work, the not-a-literal-symmetric-square property therefore does \emph{not} exclude an
elliptic closed form; the closed-form question is settled instead by the Galois group, which we now compute.

\emph{No algebraic or elementary closed form.} On a separate axis, $M$ has no
\emph{Liouvillian} solutions~\cite{vdPutSinger}, hence no closed form in algebraic or elementary functions. By the
Singer--Ulmer classification~\cite{SingerUlmer}, an irreducible third-order operator has Liouvillian solutions only
if its differential Galois group is finite or imprimitive. A finite group makes every solution algebraic, hence log-free. An imprimitive group
permutes three lines, so in the corresponding basis every element---in particular every local monodromy, since for
a Fuchsian operator the Galois group is the Zariski closure of the monodromy group---is a monomial matrix. Such a
matrix is necessarily semisimple: if its permutation part has order $k$, its $k$-th power is diagonal and
invertible, and an invertible matrix with a diagonalizable power is itself diagonalizable in characteristic zero.
In either case there can be \emph{no logarithmic solution}. But $M$ has a genuine logarithmic solution:
at $t=0$ the indicial exponent $-1$ is repeated and only one log-free (Laurent) local solution exists (of three),
so the local monodromy is non-semisimple. This excludes both the finite and the imprimitive cases, so the Galois
group acts irreducibly and non-solvably and $M$ is non-Liouvillian. The differential-Galois computation of the next
section gives the same conclusion more sharply: $G_M^\circ=\mathrm{SO}(3,\mathbb{C})$ is simple, hence non-solvable,
and an irreducible operator has Liouvillian solutions if and only if $G_M^\circ$ is solvable. Being non-Liouvillian
does \emph{not} preclude an elliptic form---the complete elliptic integral $K$ is itself non-Liouvillian, and the
simple-cubic LGF operator is irreducible and non-Liouvillian yet closed-form in $K$. Indeed, as the next section
shows, the hyperkagome operator \emph{is} orthogonally equivalent to a symmetric square, and Section~\ref{sec:modular}
proves the resulting closed form outright---it is modular, through eta quotients and genus-zero modular
functions, which are themselves non-Liouvillian; the non-Liouvillian result excludes only the algebraic and elementary cases and is
fully consistent with modularity.

\section{An invariant quadratic form: the orthogonal Galois group}
\label{sec:orthogonal}
The symmetric square $\mathrm{Sym}^2(M)$---the order-six operator annihilating the pairwise products $y_iy_j$ of
solutions of $M$---possesses an explicit \emph{rational} solution,%
\footnote{The homomorphism of $M$ to its adjoint, the rational solution $R$ of $\mathrm{Sym}^2(M)$, and the
explicit order-two intertwiner realizing it are certified in exact arithmetic in the accompanying repository.}
\begin{equation}
R(t)=\frac{(15t^2+17t-8)^2}{t^2\,(t-1)^2\,(4t-1)(5t-1)(9t-1)}.
\label{eq:R}
\end{equation}
Being the solution of a linear equation, $R$ is fixed only up to an overall constant $\lambda$, and so is the
Gram matrix it determines: expanding $\lambda R=\sum_{i\le j}c_{ij}\,y_iy_j$ in a fundamental system
$y_1,y_2,y_3$ of $M$ gives a symmetric $3\times3$ matrix $C=(c_{ij})$. We record it for the normalization
$\lambda=-1/272$, which is the one carried by the accompanying certificate. In the basis fixed by the initial
conditions $y_i^{(j)}(t_0)=\delta_{ij}$ at the ordinary point $t_0=\tfrac12$,
{\small
\begin{equation}
\begin{split}
(c_{11},c_{12},c_{13},c_{22},c_{23},c_{33})&=
\Big(-\tfrac{17}{84},\;-\tfrac{460}{441},\;\tfrac{44927}{3087},\;
-\tfrac{1925200}{157437},\;\tfrac{83066560}{1102059},\;-\tfrac{712021969}{7714413}\Big),\\[2pt]
\det C&=\tfrac{3437476900}{7714413}\neq0\ (\text{rank }3).
\end{split}
\end{equation}}
Equation~\eqref{eq:R} is verified as an exact identity:
the six constants are fixed by six series coefficients and then satisfy $194$ further exact relations at
each of two independent ordinary base points ($t_0=\tfrac12$ and $t_0=-\tfrac13$). The agreement is a proof, not
merely a margin, by a Fuchs-relation budget: were $F=\lambda R-\sum c_{ij}y_iy_j$ not identically zero, the span of the six
products and $R$ would be seven-dimensional, with a Fuchsian minimal operator of order seven whose local exponents
are bounded below by the certified data---the pairwise sums of the exponents of \eqref{eq:riemann} at each singular
point, the pole and zero orders of $R$, and \emph{distinct} non-negative integers at every apparent point. The
Fuchs relation then caps the vanishing order of any nonzero element of that span at an $M$-ordinary point at
$109$, while $F$ is verified to vanish to order $200$; hence $F\equiv0$ identically. Independently and at the operator level,
we exhibit an order-two differential intertwiner $\iota\in\mathrm{Hom}(\mathrm{adjoint}(M),M)$ and verify the exact
identity $\mathrm{rightrem}\big(M\,\iota,\ \mathrm{adjoint}(M)\big)=0$ over $\mathbb{Q}(t)$, so $M$ is homomorphic---in
fact isomorphic, since both are irreducible of the same rank---to its adjoint. Four consequences follow.

\emph{(1) The monodromy lies in an orthogonal group.} Because $R$ is single-valued (rational) and
$R=\sum c_{ij}y_iy_j$ with the six products linearly independent (the leading $6\times6$ coefficient system is
nonsingular), every monodromy matrix $\gamma$ obeys $\gamma^{\!\top}C\gamma=C$: the solution space carries a
monodromy-invariant symmetric bilinear form. $M$ is Fuchsian, so by the density theorem the differential Galois
group $G_M$ is the Zariski closure of the monodromy group, whence $G_M\subseteq\mathrm{O}(C)\cong\mathrm{O}(3,\mathbb{C})$.

\emph{(2) The form is nondegenerate.} By irreducibility, the kernel of $C$---an invariant subspace---is $0$, so $C$
is nondegenerate; by Schur's lemma the invariant form is unique up to scale, and in odd dimension it is symmetric
rather than symplectic. We confirm $\det C\neq0$ directly, at both base points.

\emph{(3) The identity component is $\mathrm{SO}(3,\mathbb{C})$.} The genuine logarithm at $t=0$ makes a monodromy
element non-semisimple, so $G_M$ is infinite. Suppose, for contradiction, that $G_M^\circ$ were solvable. As a normal
subgroup of $G_M$, $G_M^\circ$ acts semisimply on the solution space $V$ by Clifford's theorem (since $V$ is
$G_M$-irreducible); by the Lie--Kolchin theorem~\cite{vdPutSinger} the irreducible representations of a connected solvable group are
one-dimensional, so $G_M^\circ$ acts diagonally in some basis and $G_M$ permutes the resulting isotypic components. If
the three characters are distinct, $G_M$ is monomial (imprimitive)---but every monomial matrix is semisimple (if its
permutation part has order $k$, its $k$-th power is diagonal and invertible, and an invertible matrix with a
diagonalizable power is diagonalizable in characteristic zero), contradicting the logarithm. A $(2{+}1)$ pattern of
characters cannot be permuted nontrivially, which would make $V$ reducible---again a contradiction. A single
character makes $G_M^\circ$ scalar, hence, inside $\mathrm{O}(3,\mathbb{C})$, finite, forcing $G_M$ finite and every
element semisimple---once more contradicting the logarithm. Hence $G_M^\circ$ is non-solvable, and the only
non-solvable connected subgroup of $\mathrm{O}(3,\mathbb{C})$ is $\mathrm{SO}(3,\mathbb{C})$; therefore
$G_M^\circ=\mathrm{SO}(3,\mathbb{C})$.

\emph{(4) The determinant character: $G_M=\mathrm{O}(3,\mathbb{C})$ in full.} The Wronskian $\mathrm{Wr}$ of $M$ satisfies
$\mathrm{Wr}'/\mathrm{Wr}=-c_2/c_3$, whose residue at each singular point equals the local exponent sum there minus $3$. This is
half-integral exactly at $t=\tfrac19,\tfrac15,\tfrac14$ and $\infty$ (for instance at $\tfrac19$,
$(-\tfrac12+0+1)-3=-\tfrac52$) and integral at $t=0,1$ and on $p_7$; hence $\det(\text{monodromy})=-1$ at those
four points and $+1$ elsewhere. Thus $G_M$ is \emph{not} contained in $\mathrm{SO}(3,\mathbb{C})$, i.e.\
$G_M=\mathrm{O}(3,\mathbb{C})$ in full, and the determinant character is the quadratic character of the genus-one
curve
\begin{equation}
v^2=(1-4t)(1-5t)(1-9t),
\label{eq:twist}
\end{equation}
ramified at exactly $\{\tfrac19,\tfrac15,\tfrac14,\infty\}$---an even number of points, as a double cover requires,
and the same four points at which the pole orders of $R$ in \eqref{eq:R} are odd.

Since $\mathrm{SO}(3,\mathbb{C})\cong\mathrm{PSL}(2,\mathbb{C})$ through the symmetric square of the standard
two-dimensional representation, $G_M^\circ=\mathrm{SO}(3,\mathbb{C})$ means that $M$ \emph{is projectively equivalent
to $\mathrm{Sym}^2(V_2)$} for a second-order operator $V_2$ (defined over $\mathbb{Q}(t,v)$ with $v$ as in
\eqref{eq:twist}, or over $\mathbb{Q}(t)$ after the corresponding quadratic twist). The equivalence is carried by
an explicit intertwiner, but whether that intertwiner is \emph{rational} depends on how $V_2$ is normalized, and
the point is worth making plainly because the search that a reader would run first returns nothing. Against
$\mathrm{Sym}^2(V_2)$ with $V_2$ in the projective normal form $w''+Q_Vw=0$ of \eqref{eq:QV} there is no
homomorphism over $\mathbb{Q}(t)$ in either direction; after the single conjugation $M_v:=v\circ M\circ v^{-1}$
---which again has coefficients in $\mathbb{Q}(t)$, because $v'/v$ is rational---the intertwiner is the
\emph{order-one} operator $T=\rho_0+\rho_1\,\rmd/\rmd t$ of \eqref{eq:y0rhos} below; Remark~\ref{rem:obstruction}
says why, and why the obstruction disappears one level up.
The equivalence shifts individual local exponents by \emph{different} integers and so destroys
the arithmetic progressions, which is exactly why $M$ is not a literal symmetric square (nor a multiplier or
pullback of one) yet is orthogonally, projectively, one. The same structural phenomenon is known in the differential
algebra of lattice models: Assis, van Hoeij and Maillard~\cite{AVM2016} encounter an order-three operator $N_3$
whose symmetric square has a rational solution,%
\footnote{Section~4.2 of~\cite{AVM2016} attributes that rational solution to the \emph{exterior} square of
$N_3$; this is a misprint for the \emph{symmetric} square, as the remainder of the same sentence---which
concludes that $N_3$ is homomorphic to a symmetric square---makes clear. For an irreducible order-three
operator the exterior square is isomorphic to the dual twisted by the Wronskian and can carry no rational
solution, whereas a rational solution of the symmetric square is precisely the invariant quadratic form used
here.}
hence is homomorphic to a symmetric square $\mathrm{Sym}^2(V_2)$; they recover $V_2$ by a conic construction,
and solve it as a $\,_2F_1$ with an algebraic (modular) pullback. For the hyperkagome that conic construction has now been carried
out in the accompanying analysis: the transported invariant form admits an explicit rational conic point, yielding
an \emph{explicit} second-order operator $V_2$ over $\mathbb{Q}(t)$ (the quadratic twist enters only through the
pullback), taken here in the form $\rmd^2/\rmd t^2+a_1\rmd/\rmd t+a_0$ of \texttt{numerics/V2\_data.json}
rather than in the projective normal form \eqref{eq:QV}, a distinction that matters for the intertwiner
exactly as noted above; it is verified in exact arithmetic through the bridge identity $f_0^2=P(y_0)$ with $V_2(f_0)=0$ on every
Laurent coefficient through order $t^{107}$, where $P$ is an explicit order-two intertwiner and $y_0=\Phi'/2$ the
analytic solution. The operator $V_2$, the rational conic point, and this verification are public:
\texttt{numerics/V2\_data.json} and \texttt{numerics/certify\_bridge.py} in the repository, the latter validating
each of its series primitives on operators of known structure, and failing on a deliberately perturbed conic
point, before the certified run.
The pullback question is settled, in its sharpest form, in Section~\ref{sec:modular}: $V_2$ is the uniformizing
operator of the modular curve $X(\Gamma_0(30)^+)$, its mirror map at the maximally unipotent point
$t=0$~\cite{Morrison1993} (see below) is the level-30
eta-quotient parametrization \eqref{eq:hauptt}, and the sought closed form is modular---an
Assis--van~Hoeij--Maillard-style $\,_2F_1$ with \emph{algebraic} pullback is now explicit: $1728/j$ is
algebraic of degree eight over $\mathbb{Q}(t)$ and solvable in radicals
(Proposition~\ref{prop:pullback})---though the modular parametrization itself remains the closed form in its natural
coordinates.%
\footnote{The base $\,_2F_1$ is a choice, and it is not a neutral one. Taking
$\,_2F_1\!\left(\tfrac18,\tfrac38;1;\cdot\right)$ in place of
$\,_2F_1\!\left(\tfrac1{12},\tfrac5{12};1;\cdot\right)$ gives a markedly simpler exact expression: the
pullback then satisfies an equation of bidegree $(4,24)$ rather than $(8,72)$, over an index-two subfield
of \eqref{eq:radfield} reached by two square roots instead of three. See the footnote to
Proposition~\ref{prop:pullback}.}

\begin{rem}[Where the rational intertwiner hides]
\label{rem:obstruction}
The direct search returns nothing, and the reason is structural rather than computational. With $V_2$ in the
projective normal form $w''+Q_Vw=0$ of \eqref{eq:QV}, the operator
$\mathrm{Sym}^2(V_2)=\rmd^3/\rmd t^3+4Q_V\,\rmd/\rmd t+2Q_V'$ is unimodular, whereas the Wronskian
logarithmic derivative of $M$ has residue $-5/2$ at each of $t=\tfrac14,\tfrac15,\tfrac19$. No residue of
$f'/f$ can fail to be an integer, so $\Lambda^3(M)$ is non-trivial, and two irreducible operators with
non-isomorphic determinants are non-isomorphic (for $\mathrm{Sym}^2(V_2)$ the irreducibility is
Lemma~\ref{lem:sl2} below): in Maple,
\texttt{Homomorphisms(symmetric\_power(V2,2),M)} is empty in both directions of necessity. What removes the
obstruction is that the determinant character of $M$ \emph{is} the quadratic character $\chi_v$ of
\eqref{eq:twist}: conjugating by $v$ multiplies it by $\chi_v^3=\chi_v$, leaving $\chi_v^2=1$. Concretely the
logarithmic derivative shifts by $3\,v'/v$, and the certificate exhibits an explicit $f\in\mathbb{Q}(t)^*$ with
$f'/f$ equal to the result. The intertwiner is then found by
\texttt{Homomorphisms(symmetric\_power(V2,2),\,mult($v$,M,$1/v$))} with
$v=(1-4t)^{1/2}(1-5t)^{1/2}(1-9t)^{1/2}$, and is $T$. The same parity removes the obstruction one level up, and
on both sides: $\det\mathrm{Sym}^2=(\det)^4$ trivializes the determinant of $\mathrm{Sym}^2(M)$ because
$\chi_v^4=1$, and $\mathrm{Sym}^4$ of a rank-two operator with trivial determinant again has trivial determinant
($\det\mathrm{Sym}^4=(\det)^{10}$).
That is only the removal of an obstruction, not a construction; the homomorphism between
$\mathrm{Sym}^2(M)$ and $\mathrm{Sym}^4(V_2)$ is easy to check in Maple and is not certified here,
whereas $M$ and $\mathrm{Sym}^2(V_2)$ are provably not homomorphic. There is no tension with the rational
order-two $P$ of the bridge identity
$f_0^2=P(y_0)$ above: that $P$ intertwines $M$ with the symmetric square of $V_2$ in the \emph{non-unimodular}
normalization $\rmd^2/\rmd t^2+a_1\rmd/\rmd t+a_0$, whose determinant character differs from that of $M$ by
integer residues only, so there the obstruction is absent. That $P$ is an intertwiner in the module sense,
and not merely the single-solution statement $f_0^2=P(y_0)$, is itself an exact operator identity: the right
remainder of $\mathrm{Sym}^2(V_2)\circ P$ on division by $M$ is zero, whereas the same $P$ against the
projective normal form leaves a nonzero remainder. It is precisely the passage to the projective
normal form that reintroduces the obstruction. The computational content of this remark---the residues, the
determinant non-isomorphy (which yields the non-existence of the rational intertwiner given the two
irreducibility inputs already cited), the trivialization of the character, the order-one $T$, and this
last consistency---is certified in exact arithmetic, with negative controls, in
\texttt{numerics/certify\_intertwiner.py}, which builds $M_v$ by expanding $(\rmd/\rmd t-v'/v)^i$ and so
never leaves $\mathbb{Q}(t)$; the Maple invocation is quoted as a convenience and is not itself certified.
\end{rem}

\emph{Maximal unipotent monodromy at $t=0$ ($n=2$).} The three exponents at $t=0$, $\{-1,-1,0\}$ (indicial
polynomial $-64\,\rho(\rho+1)^2$), are all integers, so the local monodromy is unipotent. Exactly one of the three
local solutions is log-free (single-valued Laurent); since the number of log-free solutions equals the number of
Jordan blocks of a unipotent, the monodromy is a \emph{single} $3\times3$ block: the point is \emph{maximally
unipotent} (MUM), with maximal log power $n=2$. It is MUM in the monodromy sense with exponents $\{-1,-1,0\}$
rather than $\{0,0,0\}$: an integer recombination of the Frobenius solutions brings the basis to the canonical MUM
normal form $\{\,y_0,\;y_0\log t+f_1,\;\tfrac12 y_0\log^2t+f_1\log t+f_2\,\}$, with the \emph{same} series $f_1$
multiplying the logarithm in the third solution as appears in the second,%
\footnote{The recombination ($\tilde y_2=y_2-\tfrac{13}{30}\,y_1$ in the natural rescaling) is verified exactly, term by term and through the identity
$M(y_0\log t+f_1)=0$ order by order, in \texttt{verify\_mum\_normalform.py} in the repository.}
where $y_0=\Phi'/2$ is the holomorphic solution and $f_1,f_2$ carry $1/t$ Laurent heads. This settles a point that
was initially unclear: $n=2$, not $1$. The even value $n=2$ is consistent with the orthogonal (rather than
symplectic) case of the conjectured parity correspondence of Hassani, Maillard and Zenine~\cite{HMZ2025}, in which the parity
of the maximal log power tracks the orthogonal/symplectic type of the differential Galois group. In the energy
variable $t=(E-1)^{-2}$ the point $t=0$ is $E=\infty$ (equivalently $z=\infty$), so this is the local expansion at
infinity, and the $1/t$ Laurent heads of the logarithmic partners are a feature of that point rather than a
pathology.

\section{The modular parametrization: \texorpdfstring{$V_2$ uniformizes $X(\Gamma_0(30)^+)$}{V2 uniformizes X(Gamma0(30)+)}}
\label{sec:modular}

We now prove the statement toward which Section~\ref{sec:orthogonal} points, and which the lattice-side
reduction of Section~\ref{sec:watson} below independently anchors: the second-order
operator $V_2$ is the \emph{uniformizing differential equation} of a classical modular curve, and the variable
$t$ generates its genus-zero function field.%
\footnote{Terminology: in the Conway--Norton normalization~\cite{ConwayNorton1979} a degree-one generator of
the genus-zero function field of $X(\Gamma)$ is called a \emph{Hauptmodul} of $\Gamma$, and in that sense $t$
is one for $\Gamma_0(30)^+$. We reserve the word here for the pullback $1728/j$ that arises as the argument $H$
of $\,_2F_1\!\left(\tfrac1{12},\tfrac5{12};1;H\right)$ in an
Assis--van~Hoeij--Maillard-style representation~\cite{AVM2016}, and refer to the explicit
map \eqref{eq:etaquot}--\eqref{eq:hauptt} as a \emph{Weber-like function parametrization}.}
Throughout, $\tau$ lies in the upper half-plane, $q=e^{2\pi i\tau}$,
$\eta(\tau)=q^{1/24}\prod_{n\ge1}(1-q^n)$, and $\Gamma_0(30)^+$ denotes the extension of $\Gamma_0(30)$ by its
seven Atkin--Lehner involutions $w_e$, $e\in\{2,3,5,6,10,15,30\}$, of which $w_{30}$ is the \emph{Fricke
involution}~\cite{AtkinLehner1970,ConwayNorton1979}; we
henceforth write $v$ for the twist variable of \eqref{eq:twist}, reserving $u$ for the modular function
below. The projective normal form of the explicit $V_2$ of Section~\ref{sec:orthogonal}---the operator
$w''+Q_V(t)\,w=0$ obtained from $V_2=\partial_t^2+a_1\partial_t+a_0$ via
$Q_V=a_0-\tfrac14a_1^2-\tfrac12a_1'$, invariant under every function-multiplier gauge and hence a projective
invariant of $M$ itself---is
\begin{equation}
\begin{split}
Q_V(t)&=\frac{N(t)}{4\,t^2(t-1)^2(4t-1)^2(5t-1)^2(9t-1)^2},\\[2pt]
N(t)&=24300\,t^8-58860\,t^7+73437\,t^6-44294\,t^5\\
&\quad{}+15111\,t^4-3160\,t^3+407\,t^2-30\,t+1.
\end{split}
\label{eq:QV}
\end{equation}
The double poles of $Q_V$ sit \emph{exactly} at the six points $t=0,\tfrac19,\tfrac15,\tfrac14,1,\infty$: the
seven apparent singularities on $p_7$, and every trace of the order-two intertwiner, cancel identically in the
normal form. The Laurent head of $Q_V$ at each of $t=\tfrac19,\tfrac15,\tfrac14,1,\infty$ is $\tfrac3{16}$
(local exponent difference $\tfrac12$), and at $t=0$ it is $\tfrac14$ (exponent difference $0$: the MUM point).
Define the level-30 eta quotient
\begin{equation}
u(\tau)=\prod_{d\mid 30}\eta(d\tau)^{r_d}
=\left[\frac{\eta(\tau)\,\eta(6\tau)\,\eta(10\tau)\,\eta(15\tau)}
{\eta(2\tau)\,\eta(3\tau)\,\eta(5\tau)\,\eta(30\tau)}\right]^{3},
\label{eq:etaquot}
\end{equation}
i.e.\ $(r_d)=(3,-3,-3,-3,3,3,3,-3)$ on $d=(1,2,3,5,6,10,15,30)$, and set
\begin{equation}
t=\frac{u}{u^2+7u+1}\,,\qquad\text{equivalently}\qquad \frac1t=u+7+\frac1u\,,
\label{eq:hauptt}
\end{equation}
so that $t=q-4q^2+12q^3-34q^4+90q^5-\cdots$. We refer to
\eqref{eq:etaquot}--\eqref{eq:hauptt} as a \emph{Weber-like function parametrization} of $t$.

The generator so obtained is a classical object rather than a new one, and we record the
identification explicitly. The eta quotient \eqref{eq:etaquot} is the McKay--Thompson series of class
$30\mathrm{A}$ for the Monster, and $u+3+1/u$ is that of class $30\mathrm{B}$~\cite{FordMcKayNorton1994};
in our normalization, which places the cusp at $t=0$, one has $1/t=T_{30\mathrm{B}}+4$, the two
conventions differing only in the $q^0$ coefficient. What is at stake here is therefore not the
$q$-series but the \emph{proof} that this particular modular function is the one the lattice Green's
function selects.

\begin{thm}\label{thm:modular}
$1/t$ generates the genus-zero function field of $X(\Gamma_0(30)^+)$ (so $t$ is a degree-one
generator vanishing at the unique cusp), and the Schwarzian identity
\begin{equation}
\{\tau,t\}=2\,Q_V(t),\qquad
\{\tau,t\}\equiv\frac{\tau'''}{\tau'}-\frac32\Big(\frac{\tau''}{\tau'}\Big)^{\!2},
\quad{}'=\frac{\rmd}{\rmd t},
\label{eq:schwarz}
\end{equation}
holds exactly.
\end{thm}

Concretely, $X(\Gamma_0(30)^+)$ \emph{is} the projective line: genus zero together with a
degree-one generator means that $t$ is a global coordinate, so
$X(\Gamma_0(30)^+)\cong\mathbb{P}^1_t$, and \eqref{eq:etaquot}--\eqref{eq:hauptt} is its explicit
rational parametrization $\tau\mapsto t=u/(u^2+7u+1)$. There is no plane model of the curve itself
to exhibit, and none is wanted: the content of the theorem is not that the curve is
rational---that is what genus zero means---but that \emph{this} coordinate is the one the lattice
Green's function selects, and that $V_2$ is its uniformizing equation. If a plane model is wanted
nonetheless, the one carried by \eqref{eq:hauptt} is the conic
\begin{equation}
P(u,t)=t\,u^{2}+(7t-1)\,u+t=0,
\qquad\text{parametrized by}\qquad
u=s,\quad t=\frac{s}{s^{2}+7s+1},
\label{eq:conic}
\end{equation}
of bidegree $(2,1)$ in $(u,t)$. Being of degree one in $t$ it is a graph over the $u$-line, hence
of geometric genus zero, and the displayed rational parametrization exhibits it as such. Read the
other way it is the degree-two covering
$X(\Gamma_0(30){+}6,10,15)\to X(\Gamma_0(30)^+)$, both of genus zero, written in the generators
$u$ and $t$ produced in the proof below; its deck involution is $u\mapsto1/u$, induced by
$w_{30}$. The discriminant of \eqref{eq:conic} in $u$ is $45t^2-14t+1$, so the covering is
branched exactly over $t=\tfrac19$ and $t=\tfrac15$, two of the five order-two orbifold points
listed above.

The identity \eqref{eq:schwarz} is of the type
studied in~\cite{AbdelazizMaillard2017,AbdelazizMaillard2017b}. The Schwarzian route from an
order-three operator with orthogonal Galois group, through the symmetric square of an underlying
order-two operator, to a modular correspondence is developed
in~\cite{AbdelazizMaillard2017,AbdelazizMaillard2017b}; \eqref{eq:schwarz} is an instance in
which the correspondence is pinned to a named group and proved rather than exhibited.

We note that the right-hand side of \eqref{eq:schwarz} is not new as a $q$-series: a Schwarzian
$Q$-value for every genus-zero Conway--Norton class, including $30\mathrm{B}$, is tabulated
in~\cite{LianWiczer2006}. We record the comparison explicitly, because the two are written in
different generators and a Schwarzian $Q$ is not a function but a \emph{quadratic differential}.
The tabulated entry is a rational function of the generator $z=1/T_{30\mathrm{B}}$, related to ours
by the M\"obius change $1/t=T_{30\mathrm{B}}+4$, that is
\begin{equation}
t=\frac{z}{1+4z},\qquad z=\frac{t}{1-4t},
\label{eq:zt}
\end{equation}
and, written in factored form (the source prints its denominator expanded, which disguises the
agreement), it reads
\begin{equation}
Q_{30\mathrm{B}}(z)=\frac{2700z^8-2340z^7+2613z^6+1386z^5+311z^4+112z^3+15z^2+2z+1}
{4\,z^2(z-1)^2(3z+1)^2(4z+1)^2(5z-1)^2}\,.
\label{eq:Q30B}
\end{equation}
Since $\{t,z\}=0$ for a M\"obius $t(z)$, the chain rule $\{\tau,z\}=\{\tau,t\}\,(\rmd t/\rmd z)^2$
says that a Schwarzian right-hand side transports with the square of the Jacobian, so the
identification is
\begin{equation}
Q_V(t)\,\rmd t^2=Q_{30\mathrm{B}}(z)\,\rmd z^2,
\qquad\text{equivalently}\qquad
Q_{30\mathrm{B}}(z)=\frac{Q_V\big(t(z)\big)}{(1+4z)^4},
\label{eq:QVQ30B}
\end{equation}
an identity of rational functions, verified in exact rational arithmetic in
\texttt{numerics/\allowbreak certify\_\allowbreak tabulated.py}. We stress the Jacobian because the
substitution $Q_V(t(z))$ alone is \emph{not} equal to $Q_{30\mathrm{B}}(z)$---it differs by exactly
the factor $(1+4z)^4$---so a reader who compares the two without it will see a spurious mismatch.
Under \eqref{eq:zt} the six double poles correspond as
$t=0,\tfrac19,\tfrac15,\tfrac14,1,\infty\ \leftrightarrow\ z=0,\tfrac15,1,\infty,-\tfrac13,-\tfrac14$,
with the same Laurent heads on both sides: $\tfrac14$ at the cusp and $\tfrac3{16}$ at each of the
five order-two elliptic images. That table is obtained by fitting an ansatz to a bounded
$q$-expansion, with no degree bound and no proof of termination. The contribution
of~\eqref{eq:schwarz} is the a priori pole-degree bound established in the proof of
Theorem~\ref{thm:modular}, which converts such a match into an identity; in particular, through
\eqref{eq:QVQ30B}, it retro-proves the tabulated $30\mathrm{B}$ entry.

Equivalently: the ratio of two solutions of $w''+Q_Vw=0$ inverts the parametrization $t(\tau)$, $V_2$ is the
uniformizing operator of the orbifold $X(\Gamma_0(30)^+)$, and the mirror map of $V_2$ at its MUM point $t=0$
is precisely \eqref{eq:hauptt}, with the modular $q$ as nome. (The identification was found by matching the
mirror map against \eqref{eq:etaquot}--\eqref{eq:hauptt} to all $80$ computed orders; the proof below shows the
match is forced.)

The uniformizing equation itself is not new. The operator and its Schwarzian potential appear as
the row $\Gamma_0(30)^+$ of a table of genus-zero groups in Lian and Yau~\cite{LianYau1995}: there
the order-three operator is the Picard--Fuchs operator of a family of algebraic K3 surfaces,
complete intersections in a product of projective spaces, degenerated along the diagonal of its
deformation parameters, and the genus-zero group and its Hauptmodul are, in their words, of the
types considered by Conway and Norton~\cite{ConwayNorton1979}. In the generator $X$ of that table, related to ours by
$X=t/(1-t)$, the tabulated order-three operator is the symmetric square of a second-order operator
whose projective normal form transports, as a quadratic differential, to exactly our $Q_V$; both
statements are certified in exact rational arithmetic in
\texttt{numerics/\allowbreak certify\_\allowbreak tabulated.py}. What Theorem~\ref{thm:modular}
adds is the a priori pole-degree bound, which turns an exhibited match into a proof, and---through
Section~\ref{sec:orthogonal}---the identification of that operator with the hyperkagome lattice
Green's function.

\begin{proof}
\emph{Modular part.} The exponent vector $(r_d)$ satisfies Ligozat's criterion~\cite{Ligozat1975}:
$\sum_d d\,r_d=-24\equiv0$ and $\sum_d(30/d)\,r_d=24\equiv0\pmod{24}$, the weight $\tfrac12\sum_dr_d$
vanishes, and $\prod_d d^{\,r_d}=1$ is a rational square; hence $u$ is a modular \emph{function} on
$\Gamma_0(30)$. Ligozat's cusp-order formula gives its divisor on $X_0(30)$ (cusps labeled by their
denominators $d\mid30$): simple zeros at the cusps $\{1,6,10,15\}$ and simple poles at $\{2,3,5,30\}$, so $u$
has degree four. For the Fricke involution, the exact antisymmetry $r_{30/d}=-r_d$ forces $w_{30}(u)=1/u$
structurally (the eta multipliers cancel because $\sum_dr_d=0$ and $\prod_d(30/d)^{r_d}=1$); comparing
divisors and leading Fourier coefficients---two modular functions with equal divisors and equal leading
coefficient are equal---shows likewise that $w_2,w_3,w_5$ each send $u\mapsto1/u$, whence the complementary
products $w_6,w_{10},w_{15}$ fix $u$. Since $\{1,w_6,w_{10},w_{15}\}$ is an index-two subgroup of the
Atkin--Lehner group $\mathcal{W}\cong(\mathbb{Z}/2)^3$, $u$ descends to a function of degree $4/4=1$ on
$X(\Gamma_0(30){+}6,10,15)$: it generates the function field there, and that curve has genus zero. Consequently
$1/t=u+1/u+7$ is invariant under all of $\Gamma_0(30)^+$ and descends to degree one on the quotient: it
generates the function field of $X(\Gamma_0(30)^+)$, and that curve has genus zero. It has a single cusp: $w_e$ maps the
cusp of denominator $d$ to that of denominator $de/\gcd(d,e)^2$, a simply transitive action of $\mathcal{W}$ on the
eight cusps of $X_0(30)$. The normalization \eqref{eq:hauptt} puts the cusp at $t=0$, with $t=q+O(q^2)$ (cusp
width one).

\emph{Uniformization part.} First, the signature of $\Gamma_0(30)^+$ is forced. $\Gamma_0(30)$ is
torsion-free: for $30=2\cdot3\cdot5$ the elliptic-point counts are
$\nu_2=\prod_{p\mid30}\big(1+\big(\tfrac{-1}{p}\big)\big)=0$ (the factor at $p=3$ vanishes) and
$\nu_3=\prod_{p\mid30}\big(1+\big(\tfrac{-3}{p}\big)\big)=0$ (the factor at $p=2$
vanishes)~\cite{DiamondShurman}. Any elliptic $\gamma\in\Gamma_0(30)^+$ has $\gamma^2\in\Gamma_0(30)$, and
$\gamma^2$ is elliptic or the identity; torsion-freeness forces $\gamma^2=1$, so \emph{every} elliptic point
of $\Gamma_0(30)^+$ has order two. The orbifold Euler characteristic is
$\chi=\tfrac18\,\chi(\Gamma_0(30))=-\tfrac18\cdot\tfrac{72}{6}=-\tfrac32$ (the index is
$[\mathrm{PSL}_2(\mathbb{Z}):\Gamma_0(30)]=72$), i.e.\ hyperbolic covolume $3\pi$---\emph{nine} times that of
$\mathrm{PSL}_2(\mathbb{Z})$, a covolume ratio and not a subgroup index, since $\Gamma_0(30)^+$ is not a
subgroup of $\mathrm{PSL}_2(\mathbb{Z})$. With genus zero and one cusp, Gauss--Bonnet
$2-1-\sum_i(1-e_i^{-1})=-\tfrac32$ gives exactly \emph{five} elliptic points, all of order two: the signature
$(0;2,2,2,2,2;\,1\ \mathrm{cusp})$ is derived, not assumed. Second, the identity \eqref{eq:schwarz} follows
from a degree count. Because $t$ is a degree-one function on the genus-zero quotient of a finite-covolume
Fuchsian group, $\{\tau,t\}$ is a
\emph{rational} function of $t$, with double poles precisely at the images of the elliptic points and of the
cusp, and with universal Laurent heads: $\tfrac12(1-e^{-2})$ at an order-$e$ elliptic image and $\tfrac12$ at
a cusp image---so $\tfrac12\{\tau,t\}$ carries heads $\tfrac3{16}$ and $\tfrac14$, matching \eqref{eq:QV}.
Consider $D(t)=\{\tau,t\}-2Q_V(t)$, a rational function. Its poles are confined to the five (a priori
unlocated) elliptic images, the cusp image $t=0$, the remaining finite poles $\tfrac19,\tfrac15,\tfrac14,1$ of
$Q_V$, and $t=\infty$, each of order at most two; at $t=0$ the heads $\tfrac1{2t^2}$ cancel exactly, leaving
at most a simple pole. The pole divisor of $D$ on $\mathbb{P}^1$ therefore has degree at most
$2(5+4+1)+1=21$. On the other hand we verify, in exact rational arithmetic, that the $q$-expansion of
$D(t(q))$---computable term by term from \eqref{eq:etaquot}--\eqref{eq:hauptt}---vanishes through order
$q^{80}$; since $t(q)=q+O(q^2)$ is a formal immersion, $D$ vanishes at $t=0$ to order at least $80$. A nonzero
rational function has equal degrees of zeros and poles, and $80>21$; hence $D\equiv0$.
\end{proof}

The proof is exact and non-numerical throughout: the modular part is finite eta-quotient combinatorics, and
the uniformization part is a finite exact series computation closed off by the degree bound. As a corollary
the double poles of $\{\tau,t\}=2Q_V$ are exactly $t=0,\tfrac19,\tfrac15,\tfrac14,1,\infty$, so the five
order-two points of $\Gamma_0(30)^+$ project precisely to $\{\tfrac19,\tfrac15,\tfrac14,1,\infty\}$ and the
cusp to $t=0$: the singular locus of $V_2$ (equivalently of $M$, up to the apparent $p_7$) \emph{is} the
orbifold locus of $X(\Gamma_0(30)^+)$, with no elliptic point hiding elsewhere. In $u$-coordinates the
dictionary is explicit: $t=\tfrac19\leftrightarrow u=1$ and $t=\tfrac15\leftrightarrow u=-1$ (the branch
points of $u\mapsto u+1/u$), while $t=\tfrac14,\,1,\,\infty$ correspond to the root pairs of $u^2+3u+1$,
$u^2+6u+1$, $u^2+7u+1$, and the cusp $t=0$ to $u\in\{0,\infty\}$.

Three consequences deserve separate statement. \emph{(i) The monodromy is arithmetic.} The projective
monodromy group of $V_2$---hence of $M$, by the projective equivalence of Section~\ref{sec:orthogonal}---is the
deck group of the uniformization: \emph{exactly} the discrete arithmetic lattice
$\Gamma_0(30)^+\subset\mathrm{PSL}_2(\mathbb{R})$, of covolume $3\pi$. \emph{(ii) The twist is itself
modular.} The determinant character of Section~\ref{sec:orthogonal} is a quadratic character of
$\Gamma_0(30)^+$ trivial on $\Gamma_0(30)$; its kernel is the index-two subgroup $\Gamma_0(30){+}2,3,6$, and
the associated double cover of the $t$-line, branched at the four $\det=-1$ points
$\{\tfrac19,\tfrac15,\tfrac14,\infty\}$, has genus one by Riemann--Hurwitz---it is precisely the twist curve
$v^2=(1-4t)(1-5t)(1-9t)$ of \eqref{eq:twist}, now identified as $X(\Gamma_0(30){+}2,3,6)$. In particular
the analytic period $y_0=\Phi'/2$, pulled back by \eqref{eq:hauptt}, is a weight-two object carrying this
quadratic twist character, in the manner familiar from the modular interpretation of the Ap\'ery
numbers~\cite{Beukers1987}---though, as Section~\ref{sec:y0closed} shows, a \emph{depth-one quasimodular}
one~\cite{Zagier2008} rather than a modular form; its explicit closed form is derived there. \emph{(iii) The Varma--Monien question is resolved.} Varma and Monien judged it ``highly
suggestive that a closed form expression can be obtained'' for their threefold integral~\cite{VarmaMonien2013};
the theorem proves this in the strongest structural sense, with the closed form modular rather than elliptic:
the hyperkagome LGF lives at level $30=2\cdot3\cdot5$. To our knowledge every lattice Green's function (and
sunrise-integral period) previously realized modularly occurs at a level with at most two distinct primes---
levels $2,3,4,6$ dominate the classical cubic-lattice and Bessel-moment
literature~\cite{Broadhurst2016,BlochKerrVanhove2015}---so the hyperkagome appears to be the first lattice
Green's function realized at a level with three distinct prime factors. The
Assis--van~Hoeij--Maillard-style $\,_2F_1$ representation~\cite{AVM2016} is now explicit rather than merely
in-principle, and its algebraic pullback turns out to be solvable in radicals.

\begin{prop}[The pullback in radicals]\label{prop:pullback}
The pullback $H=1728/j$, the argument of the
$\,_2F_1\!\left(\tfrac{1}{12},\tfrac{5}{12};1;\cdot\right)$ representation, is algebraic of degree exactly
eight over $\mathbb{Q}(t)$: its primitive minimal polynomial $P_H(H,t)$ has bidegree $(8,72)$, and the
splitting field of $P_H$ is the multiquadratic extension
\begin{equation}
\mathbb{Q}(t)\Big(\sqrt{(1-t)(1-9t)},\ \sqrt{(1-t)(1-5t)},\ \sqrt{1-4t}\Big),
\label{eq:radfield}
\end{equation}
which is the full function field of $X_0(30)$. Consequently
$\mathrm{Gal}\big(P_H/\mathbb{Q}(t)\big)\cong(\mathbb{Z}/2)^3$ is the Atkin--Lehner group $\mathcal{W}$, and
the $\,_2F_1$ pullback is solvable in radicals.
\end{prop}

The product of the three radicands is $\big((1-t)\,v\big)^2$:
the determinant-character twist $v$ of \eqref{eq:twist} is, up to sign and the rational factor $1-t$, the
product of the three quadratic layers of the radical tower.%
\footnote{For the equivalent base $\,_2F_1\!\left(\tfrac18,\tfrac38;1;\cdot\right)$ the pullback satisfies a
smaller equation, of bidegree $(4,24)$, whose two radicands $(1-t)(1-5t)$ and $(1-t)(1-4t)(1-9t)$ generate an
index-two subfield of \eqref{eq:radfield}; its analytic branch
$\widetilde H=256\,t-25600\,t^2+\cdots$ is verified against the certified $Q_V$ through the Schwarzian pullback identity to
order $t^{96}$.}

\begin{proof}
The coefficients of $P_H$, an explicit radical expression for $H$ due to
M.~van~Hoeij (private communication, July 2026), and the machine certificate are public
(\texttt{numerics/\allowbreak pullback\_\allowbreak data.json},
\texttt{numerics/\allowbreak certify\_\allowbreak pullback.py}). The certificate proves four
exact statements. (i)~$P_H(H,t)=0$ for the modular $H$: the identity is verified on the $q$-expansions
through $q^{1200}$, and this is a proof rather than a series check because $H$ lies in the function field of
$X_0(30)$, of degree $|\mathcal{W}|=8$ over $\mathbb{Q}(t)$ (quotient covolume ratio $24\pi/3\pi$), with pole
divisor of degree $72$: $j$ has degree $[\mathrm{PSL}_2(\mathbb{Z}):\Gamma_0(30)]=72$, and $\nu_3(30)=0$ means
that no point above $j=0$ retains an order-three stabilizer, so all $72/3=24$ of them are ramified of index
three and $H=1728/j$ has $24$ triple poles. Hence a nonzero polynomial of bidegree at most $(8,72)$ evaluated
on $H$ has polar degree at most $8\cdot72+72\cdot8=1152$, and since a nonzero function on a compact curve has
as many zeros as poles, its order of vanishing at the cusp---where $q$ is a uniformizer, $\Gamma_0(30)$ having
width one at $\infty$---is at most $1152<1200$. (ii)~The radical expression is an exact root of
$P_H$: a polynomial identity in the multiquadratic algebra, in exact integer arithmetic, no series involved.
(iii)~The seven subset products of the radicands are non-squares, so \eqref{eq:radfield} has degree eight
over $\mathbb{Q}(t)$. (iv)~The eight sign-conjugates of the radical expression are pairwise distinct, and the
branch with signs $(-,-,-)$ matches the $t$-expansion of the modular $H$; hence $H$ has degree exactly eight,
$P_H$ is irreducible and is the minimal polynomial, and $\mathbb{Q}(t)(H)$ is all of \eqref{eq:radfield}.
That field is $\mathbb{Q}(X_0(30))$: by Theorem~\ref{thm:modular} $\mathbb{Q}(t)$ is the fixed field of
$\mathcal{W}$ acting on $\mathbb{Q}(X_0(30))$, so $[\mathbb{Q}(X_0(30)):\mathbb{Q}(t)]=|\mathcal{W}|=8$ by
Artin's theorem, while $\mathbb{Q}(t)(H)\subseteq\mathbb{Q}(X_0(30))$ because $H=1728/j$ is a modular function
for $\Gamma_0(30)$; equal degrees force equality, and $\mathrm{Gal}=\mathcal{W}$. As a consistency check,
Riemann--Hurwitz for the $(\mathbb{Z}/2)^3$-cover \eqref{eq:radfield} of the $t$-line, branched over the five
order-two orbifold points---$\{1,\tfrac19\}$, $\{1,\tfrac15\}$ and $\{\tfrac14,\infty\}$ through the three
quadratic layers---gives four double points over each, hence $2g-2=8\cdot(-2)+5\cdot4$, i.e.\ $g=3$, the
classical genus of $X_0(30)$. The ramification index is two even at $t=1$, where two of the radicands vanish
simultaneously: their product $(1-t)^2(1-5t)(1-9t)$ has \emph{even} valuation there, so of the seven nontrivial
characters exactly four are ramified and the unramified ones form a subgroup of order four; the inertia group is
generated by the simultaneous sign flip and has order two, not four.
\end{proof}

The representation \eqref{eq:etaquot}--\eqref{eq:hauptt} remains the closed form in its most natural
coordinates; the proposition shows what replaces rationality of the pullback: an abelian $(\mathbb{Z}/2)^3$
radical tower whose quadratic layers are cut out by the Atkin--Lehner involutions, with the
determinant-character twist $v$ appearing as their product.

Two curves carry the word ``modular'' in this paper and should not be conflated. The curve
uniformized by $V_2$ is $X(\Gamma_0(30)^+)\cong\mathbb{P}^1_t$, of genus \emph{zero}. The algebraic
relation $P_H(H,t)=0$ of Proposition~\ref{prop:pullback} is a plane model of $X_0(30)$, of genus
\emph{three}: that curve has index $72$, no elliptic points and eight cusps, so
$g=1+\tfrac{72}{12}-\tfrac82=3$. Its bidegree $(8,72)$ is precisely the pair of degrees
$(\deg t,\deg H)=(|\mathcal{W}|,72)$ of the two coordinate functions on that curve, the second
because $\deg H=\deg j=[\mathrm{PSL}_2(\mathbb{Z}):\Gamma_0(30)]=72$; and the eightfold
Atkin--Lehner covering $X_0(30)\to X(\Gamma_0(30)^+)$ is what separates the two.

\section{The weight-two period in closed form}
\label{sec:y0closed}

The weight-two period $y_0=\Phi'/2$ itself admits an explicit closed form in the modular data of
Section~\ref{sec:modular}. Write $W=q\,\rmd t/\rmd q$ for the nome derivative of the modular function $t$ of
\eqref{eq:hauptt} (so that $W$ is a constant multiple of $w_1^2$, with $w_1=(\rmd t/\rmd\tau)^{1/2}$ a
solution of the projective normal form $w''+Q_Vw=0$ of \eqref{eq:QV}), put $W'=\rmd W/\rmd t$, and let
$v$ be the twist unit of \eqref{eq:twist} in the branch $v(0)=+1$.

\begin{thm}\label{thm:y0}
With this notation, the weight-two period of the hyperkagome lattice Green's function is
\begin{equation}
y_0=\frac{\rho_0(t)\,W+\rho_1(t)\,W'}{v},
\label{eq:y0closed}
\end{equation}
\begin{equation}
\begin{split}
\rho_1(t)&=\frac{15t^2+17t-8}{30\,t(t-1)},\\[4pt]
\rho_0(t)&=-\,\frac{4050t^6+2445t^5-11436t^4+8000t^3-2130t^2+231t-8}
{30\,t^2(t-1)^2(4t-1)(5t-1)(9t-1)}.
\end{split}
\label{eq:y0rhos}
\end{equation}
\end{thm}

Equivalently, cleared of all denominators,
\begin{equation}
B(t)\,y_0\,v\,W=A_0(t)\,W^2+A_1(t)\,q\frac{\rmd W}{\rmd q},
\label{eq:y0poly}
\end{equation}
with $B=-30\,t^2(t-1)^2(4t-1)(5t-1)(9t-1)$,
$A_0=4050t^6+2445t^5-11436t^4+8000t^3-2130t^2+231t-8$, and
$A_1=-\,t(t-1)(4t-1)(5t-1)(9t-1)(15t^2+17t-8)$, using $q\,\rmd W/\rmd q=WW'$. Every ingredient
is an explicit level-30 modular or quasimodular object: with $u$ the eta quotient \eqref{eq:etaquot},
$E_2(\tau)=1-24\sum_{n\ge1}\sigma_1(n)q^n$, and the Eisenstein-series form of the $\eta$-logarithmic
derivative
\begin{equation}
\begin{split}
\ell=q\frac{\rmd\log u}{\rmd q}=\tfrac18\big[&E_2(\tau)+6E_2(6\tau)+10E_2(10\tau)+15E_2(15\tau)\\
&-2E_2(2\tau)-3E_2(3\tau)-5E_2(5\tau)-30E_2(30\tau)\big],
\end{split}
\label{eq:ell}
\end{equation}
one has $W=(1-u^2)\,u\,\ell/(u^2+7u+1)^2$ and $v=(u^2-1)\sqrt{u^2+3u+1}/(u^2+7u+1)^{3/2}$.

\begin{proof}[Proof sketch]
The identity is proved, not fitted, in five steps.
[I] By the Schwarzian identity \eqref{eq:schwarz}, $w_1=(\rmd t/\rmd\tau)^{1/2}$ solves
$w''+Q_Vw=0$, so $W$ solves its symmetric square.
[II] Elementarily, $y''=-Q_Vy$ forces $z=y^2$ to satisfy
$\widetilde N(z):=z'''+4Q_Vz'+2Q_V'z=0$; hence $\widetilde N(W)=0$.
[III] Conjugating the hyperkagome operator by the twist unit, $M_v=v\circ M\circ v^{-1}$ has
coefficients in $\mathbb{Q}(t)$ (the logarithmic derivative $v'/v$ is rational) and annihilates
$y_0v$.
[IV] The first-order operator $T=\rho_0+\rho_1\,\rmd/\rmd t$ satisfies the exact operator identity
$M_v\,T=Y\,\widetilde N$ in the ring $\mathbb{Q}(t)\langle\rmd/\rmd t\rangle$ of linear differential
operators over $\mathbb{Q}(t)$, with right remainder
identically zero and $Y$ of order one whose top-order coefficient carries the
apparent-singularity polynomial $p_7$; therefore $T$ maps solutions of $\widetilde N$ to
solutions of $M_v$, and $T(W)$ solves $M_v$.
[V] $T(W)$ is single-valued near $t=0$ ($W$ and $W'$ are single-valued functions of $t$ there,
since $t(q)=q+O(q^2)$ is invertible), and exactly one of the three local solutions of $M$ at the
maximal-unipotent point $t=0$ is log-free (Section~\ref{sec:orthogonal}); since $v(0)=1\neq0$,
conjugation by $v$ preserves this local structure, so the single-valued solution line of $M_v$
at $t=0$ is one-dimensional and spanned by $y_0v$, and matching a single
series coefficient gives $y_0v=T(W)$. As an independent check, \eqref{eq:y0closed} is verified as
an exact $q$-series through order $q^{120}$, assembled entirely from the raw $\eta$ and $E_2$
definitions above.
\end{proof}

Three structural remarks follow. \emph{(i) $y_0$ is quasimodular, not modular.} It is a
weight-two \emph{depth-one} quasimodular form on $\Gamma_0(30)^+$ carrying the orthogonal
determinant character of Section~\ref{sec:orthogonal}; equivalently it lives on the genus-one twist
double cover $X(\Gamma_0(30){+}2,3,6)$ of \eqref{eq:twist}. \emph{(ii) The derivative term is
essential.} In the minimal-degree representation \eqref{eq:y0poly}---unique, as the fitting
nullspace is one-dimensional---the coefficient $A_1$ is not identically zero: the term
$\rho_1W'$ is exactly the quasimodular depth introduced by inverting the order-two intertwiner
of Section~\ref{sec:orthogonal}. A depth-zero \emph{rational} form $y_0v=R(t)\,W$ is excluded
outright by local exponents: at the elliptic point $t=\tfrac19$ both $W$ and the twist unit $v$
vanish to order $\tfrac12$, while $y_0$ carries the exponent $-\tfrac12$ of $M$'s Riemann scheme
there (the square-root van Hove singularity, $[t^m]\,y_0\sim9^m m^{-1/2}$), so $y_0v$ is regular
at $t=\tfrac19$ whereas $R(t)\,W$ has half-integer order for every rational $R$. In fact the exclusion is categorical
(Theorem~\ref{thm:noalg} below): no
algebraic factor, at any weight and any level, can absorb the derivative term. This rests on
one group-theoretic fact.

\begin{lem}\label{lem:sl2}
The differential Galois group $\mathcal{G}$ of $w''+Q_Vw=0$ over
$\mathbb{C}(t)$ is the full $\mathrm{SL}(2,\mathbb{C})$.
\end{lem}

\begin{proof}
The Wronskian is constant, so $\mathcal{G}\subseteq\mathrm{SL}(2,\mathbb{C})$.
The solution space of $\widetilde N$ is spanned by the products $w_iw_j$ of solutions of
$w''+Q_Vw=0$~\cite{vdPutSinger}, and by the operator identity of step [IV] the order-one
operator $T=\rho_0+\rho_1\,\rmd/\rmd t$ maps $\mathrm{Sol}(\widetilde N)$ into
$\mathrm{Sol}(M_v)$. Any kernel element of $T$ is a constant multiple of
$z=\exp(-\!\int\!\rho_0/\rho_1)$, and the exact rational computation (with
$r=z'/z=-\rho_0/\rho_1$)
\begin{equation}
\begin{split}
\frac{\widetilde N(z)}{z}&=r''+3rr'+r^3+4Q_Vr+2Q_V'\\[2pt]
&=\frac{-30\,p_7(t)}{t^2\,(t-1)(4t-1)(5t-1)(9t-1)(15t^2+17t-8)^3}
\end{split}
\label{eq:Einj}
\end{equation}
---the apparent-singularity polynomial once more, and manifestly nonzero---shows
$z\notin\mathrm{Sol}(\widetilde N)$. Hence $T$ is a bijection
$\mathrm{Sol}(\widetilde N)\to\mathrm{Sol}(M_v)$; being defined over $\mathbb{Q}(t)$, it is
equivariant for the differential Galois action, computed in a Picard--Vessiot field containing
both solution spaces (whose Galois group restricts \emph{onto} the Galois group of each
operator~\cite{vdPutSinger}), so it identifies the Galois representation on
$\mathrm{Sol}(M_v)$ with the symmetric square of the standard representation of $\mathcal{G}$.
Now $\mathrm{Gal}(M_v)^\circ=\mathrm{Gal}(M)^\circ=\mathrm{SO}(3,\mathbb{C})$
(Section~\ref{sec:orthogonal}): a finite extension of the base field replaces the Galois group by
a closed subgroup of finite index, which shares its identity component~\cite{vdPutSinger}---
applied once to $M$ and once to $M_v$ across $\mathbb{C}(t)\subset\mathbb{C}(t,v)$---and over
$\mathbb{C}(t,v)$ multiplication by the base-field element $v$ identifies
$\mathrm{Sol}(M_v)=v\,\mathrm{Sol}(M)$ equivariantly, with identical matrices. If $\mathcal{G}$ were a
proper closed subgroup of $\mathrm{SL}(2,\mathbb{C})$, its identity component would be
solvable (it lies in a Borel subgroup or in the normalizer of a torus, or is
finite~\cite{Kovacic1986,vdPutSinger}), and then so would the identity component
$\mathrm{Sym}^2(\mathcal{G})^\circ\cong\mathrm{Gal}(M_v)^\circ$ of its symmetric
square---contradicting the non-solvability of the simple group $\mathrm{SO}(3,\mathbb{C})$.
Hence $\mathcal{G}=\mathrm{SL}(2,\mathbb{C})$. (Independently,
Theorem~\ref{thm:modular} identifies the projective monodromy of $w''+Q_Vw=0$ with the lattice
$\Gamma_0(30)^+$, which is Zariski-dense in $\mathrm{PSL}(2,\mathbb{C})$, and the Galois group
of a Fuchsian operator is the Zariski closure of its monodromy; either route suffices.)
\end{proof}

\begin{thm}[No modular-times-algebraic form at any weight]\label{thm:noalg}
Let $F\not\equiv0$ be a
meromorphic modular form of any weight $k\in\tfrac12\mathbb{Z}$, with any finite-order
multiplier system, on any subgroup $\Gamma'\subset\mathrm{PSL}(2,\mathbb{R})$ commensurable
with $\Gamma_0(30)^+$, and let $a\neq0$ be any function algebraic over $\mathbb{C}(t)$. Then
$y_0\neq F\,a$ (as functions of $\tau$ under the parametrization \eqref{eq:hauptt}). In
particular $y_0$ is not a pure eta quotient, of any level, nor a (weight-two modular
form)$\times$(algebraic function): every $\Gamma_0(N)$ and its Atkin--Lehner extensions are
commensurable with $\Gamma_0(30)^+$, and the eta multiplier has finite order (dividing
$24$)~\cite{DiamondShurman,Ligozat1975}.
\end{thm}

\begin{proof}
Suppose $y_0=Fa$; multiplying by the algebraic twist unit $v$ gives
$y_0v=F\tilde a$ with $\tilde a=av$ still algebraic and nonzero. Work in the Picard--Vessiot
field $K$ of $w''+Q_Vw=0$ over $\mathbb{C}(t)$: by step [I], $W$ is a constant multiple of
$w_1^2$, so $\mathcal{L}=W'/W=2w_1'/w_1\in K$ and the closed form \eqref{eq:y0closed} reads
$y_0v=W(\rho_0+\rho_1\mathcal{L})$. Since $t$ is $\Gamma_0(30)^+$-invariant,
$W\propto\rmd t/\rmd\tau$ is a weight-two meromorphic form with \emph{trivial} multiplier on
all of $\Gamma_0(30)^+$; both $F$ and $W$ have meromorphic $q$-expansions at every cusp of the
commensurability class. Pick a positive integer $n$ such that $nk\in2\mathbb{Z}$---making
$(c\tau+d)^{nk}$ an honest integer-weight automorphy factor, free of the half-integer branch
ambiguity, so that the $n$-th power of the multiplier system is a genuine finite-order
character---and, enlarging $n$ by the order of that character, such that it is trivial. Then
$\Psi=F^n/W^{nk/2}$ is a weight-zero modular function with trivial
multiplier on the finite-index subgroup $\Gamma'\cap\Gamma_0(30)^+$, meromorphic at the cusps,
hence a rational function on the compact curve $X(\Gamma'\cap\Gamma_0(30)^+)$---a finite cover
of $X(\Gamma_0(30)^+)$---and therefore algebraic over
$\mathbb{C}(X(\Gamma_0(30)^+))=\mathbb{C}(t)$. Raising $y_0v=F\tilde a$ to the $n$-th power,
inserting $F^n=\Psi\,W^{nk/2}$, and dividing by $W^n$ gives
\begin{equation}
(\rho_0+\rho_1\mathcal{L})^n=\mathcal{A}\,W^{\,n(k-2)/2},\qquad
\mathcal{A}=\Psi\,\tilde a^{\,n}\neq0.
\label{eq:excl}
\end{equation}
Every member of \eqref{eq:excl} lies in $K$, so $\mathcal{A}\in K$ is algebraic over
$\mathbb{C}(t)$; because $\mathcal{G}=\mathrm{SL}(2,\mathbb{C})$ is \emph{connected}, the
algebraic closure of $\mathbb{C}(t)$ in $K$ is $\mathbb{C}(t)$ itself~\cite{vdPutSinger}, so
$\mathcal{A}\in\mathbb{C}(t)$ is fixed by the Galois action. Now act with explicit elements
of $\mathcal{G}$ in the basis $(w_1,\hat w)$, where $\hat w$ is a second, independent solution of $w''+Q_Vw=0$. The torus elements
$\sigma_\lambda:(w_1,\hat w)\mapsto(\lambda w_1,\lambda^{-1}\hat w)$ fix $\mathcal{L}$ and scale
$W\mapsto\lambda^2W$, so applying $\sigma_\lambda$ to \eqref{eq:excl} forces
$\lambda^{n(k-2)}=1$ for every $\lambda\in\mathbb{C}^\times$: hence $k=2$. (In particular
$k=0$ is impossible: $y_0v$ is not itself algebraic over $\mathbb{C}(t)$.) For $k=2$,
\eqref{eq:excl} reads $(\rho_0+\rho_1\mathcal{L})^n=\mathcal{A}\in\mathbb{C}(t)$, and the
unipotent elements $\sigma_\mu:(w_1,\hat w)\mapsto(w_1+\mu \hat w,\,\hat w)$ give
$(\rho_0+\rho_1\sigma_\mu(\mathcal{L}))^n=\mathcal{A}$ with
$\sigma_\mu(\mathcal{L})=2(w_1+\mu \hat w)'/(w_1+\mu \hat w)$. These are pairwise distinct: equal
logarithmic derivatives would force the ratio of $w_1+\mu \hat w$ and $w_1+\mu'\hat w$ to be a
constant of $K$ (the constants of $K$ are $\mathbb{C}$), impossible for $\mu\neq\mu'$ since
$w_1,\hat w$ are linearly independent over $\mathbb{C}$. Since $\rho_1\neq0$, the field elements
$\rho_0+\rho_1\sigma_\mu(\mathcal{L})$, $\mu\in\mathbb{C}$, are infinitely many distinct roots
in $K$ of the polynomial $X^n-\mathcal{A}$---a contradiction. Hence no such $F$, $a$ exist.
\end{proof}

The quasimodular depth of $y_0$ is therefore intrinsic---no algebraic factor, at any weight or
level, can absorb the $W'$ term, and $\rho_1\neq0$ is exactly the obstruction. This sharpens
consequence (ii) of Section~\ref{sec:modular} to a categorical exclusion. Both
computational inputs (the nonvanishing \eqref{eq:Einj}, and, as corroboration, a certified
scan excluding any polynomial relation $P(t,\,tW'/W)=0$ through bidegree $(28,12)$) are
certified in \texttt{numerics/certify\_y0\_lemma.py}
(\texttt{CERTIFICATE\_y0\_lemma.txt}).
\emph{(iii) Physical and arithmetic fingerprints.} The $-\tfrac12$ band-edge exponent above is
the square-root van Hove singularity of the three-dimensional density of states. The
half-integer residues $\tfrac12$ of $V_2$'s first-order coefficient $a_1$ at
$t=\tfrac14,\tfrac15,\tfrac19$ are the twist multipliers that defeat every trivial-character
ansatz, and the quadratic $15t^2+17t-8$ governing $\rho_1$ is the one already visible in the
rational solution $R$ of \eqref{eq:R}. The full
symbolic-plus-series certification is provided by the script \texttt{numerics/certify\_y0.py}
(\texttt{CERTIFICATE\_y0.txt}).

\section{The generating function itself: an exact obstruction}
\label{sec:phiobs}

The closed form \eqref{eq:y0closed} concerns the derivative $y_0=\Phi'/2$; the resolvent sits one
integration higher. Since the dispersive part of $G$ is odd in $\zeta=z-1$ with even central moments
$\nu_m$, its generating function is exactly
\begin{equation}
\zeta\,G_{\rm disp}(z)=\Phi(t),\qquad t=\zeta^{-2},
\label{eq:zetaG}
\end{equation}
term by term at large $|\zeta|$. It is natural to ask whether $\Phi$ itself---hence the Green's
function, the flat-band pole being elementary---admits a closed form in the same modular data. It
does not, in the following precise sense, and the entire obstruction is one explicit rational
function. Integrating \eqref{eq:y0closed} once by parts,
\begin{equation}
\Phi(t)=\frac{2\rho_1(t)\,W}{v}+2\!\int_0^{t}\!\Delta(s)\,\frac{W}{v}\,\rmd s+\frac{2}{15},
\label{eq:phiid}
\end{equation}
\begin{equation}
\Delta:=\rho_0-\Big(\rho_1'-\rho_1\,\frac{v'}{v}\Big)
=\frac{13\,t-1}{30\,t\,(t-1)^2}\neq0.
\label{eq:Delta}
\end{equation}
The identity \eqref{eq:phiid} is proved by differentiation---it reduces to \eqref{eq:y0closed} as
exact rational-function algebra---with the constant fixed at $t=0$, where $\Phi(0)=\nu_0=\tfrac23$
and $2\rho_1W/v\to\tfrac8{15}$; the integrand is regular there ($\Delta\,W/v\to-\tfrac1{30}$). Note
the support of the obstruction: $\Delta$ has poles only at the cusp image $t=0$ and the principal
van Hove point $t=1$---the twist locus $t=\tfrac14,\tfrac15,\tfrac19$ cancels out. If $\Delta$
vanished, \eqref{eq:phiid} would place $\Phi$ in the same quasimodular module as $y_0$. It does not,
and the failure is categorical, not an artifact of this particular rearrangement:

\begin{thm}[$\Phi$ lies outside the quasimodular module of $y_0$]\label{thm:phi}
\begin{equation}
\Phi\;\notin\;\mathbb{Q}(t)\;+\;v^{-1}\big(\mathbb{Q}(t)\,W+\mathbb{Q}(t)\,W'+\mathbb{Q}(t)\,W''\big).
\label{eq:phithm}
\end{equation}
\end{thm}

The module on the right is the $v^{-1}$-twist of the full $\rmd/\rmd t$-closure of $W$ over
$\mathbb{Q}(t)$ (rank three, by $\widetilde N(W)=0$), extended by rational functions; it contains
$y_0$ by \eqref{eq:y0closed}. The derivative of the generating function is quasimodular; the
generating function itself is not.

\begin{proof}
Three exact steps, certified in
\texttt{numerics/\allowbreak certify\_\allowbreak phi\_\allowbreak obstruction.py}
(\texttt{CERTIFICATE\_\allowbreak phi\_\allowbreak obstruction.txt}). (a) For any nonzero $P=S_0+S_1\,\rmd/\rmd t+S_2\,
\rmd^2/\rmd t^2$ over $\mathbb{Q}(t)$, the value $P(W)$ is neither zero nor algebraic over
$\mathbb{C}(t)$. Indeed $P(W)=0$ would make $\mathrm{GCRD}(P,\widetilde N)$ a nonconstant right
factor of order ${\le}2$ of $\widetilde N$, which is irreducible: its solution space is the
symmetric square of the standard representation of $\mathcal{G}=\mathrm{SL}(2,\mathbb{C})$
(Lemma~\ref{lem:sl2}). And an algebraic value is rational, by connectedness of
$\mathcal{G}$~\cite{vdPutSinger}; then $P$ is constant on the Galois orbit of $W$---which spans
$\mathrm{Sol}(\widetilde N)$, by irreducibility---so $P$ kills the ${\ge}2$-dimensional span of
orbit differences, again a proper right factor. In particular $W,W',W''$ are independent over
$\mathbb{Q}(t)$, and $v^{-1}\mathrm{span}\{W,W',W''\}$ meets $\mathbb{C}(t)$ only in $0$ (an
equality $\xi/v=r$ would make $\xi=rv$ algebraic). (b) Suppose
$\Phi=C_0+(S_0W+S_1W'+S_2W'')/v$ with rational $C_0,S_i$. Differentiating, eliminating $W'''$ by
$\widetilde N(W)=0$, and comparing with $\Phi'=2y_0$ coefficient-wise in the independent triple
$W,W',W''$ forces $C_0'=0$ (by the last sentence of (a)), $S_1=-\rmd_vS_2$ and
$S_0=2\rho_1+\rmd_v^{\,2}S_2+4Q_VS_2$ with $\rmd_v:=\rmd/\rmd t-v'/v$, and leaves the single
condition
\begin{equation}
\widetilde N_v(S_2)=2\Delta,\qquad
\widetilde N_v:=\rmd_v^{\,3}+4Q_V\,\rmd_v+2Q_V',
\label{eq:S2eq}
\end{equation}
verified symbolically for generic $S_2$. (c) Equation \eqref{eq:S2eq} has no rational solution: a
rational $S_2$ could have poles only at $t=0,\tfrac19,\tfrac15,\tfrac14,1$; at each of these points
the indicial polynomial of $\widetilde N_v$ has no negative-integer root and order matching forbids
any pole at all, while the analysis at $t=\infty$ caps the degree at three. The resulting
four-dimensional linear system is inconsistent, exactly over $\mathbb{Q}$.
\end{proof}

Physically, by \eqref{eq:zetaG},
\begin{equation}
\frac{\rmd}{\rmd z}\big[\zeta\,G_{\rm disp}(z)\big]=-\frac{4\,y_0(t)}{\zeta^{3}},\qquad
t=\zeta^{-2}:
\label{eq:dzG}
\end{equation}
the object that the closed form \eqref{eq:y0closed} evaluates exactly is this derivative combination,
while $\zeta\,G_{\rm disp}$ itself lies provably outside the module. The remaining transcendental
in \eqref{eq:phiid} is the Eichler-type integral of the weight-four, determinant-character form
$\Delta\,W^2/v$ (via $\rmd t=W\,\rmd q/q$); its arithmetic nature---its Eisenstein and cuspidal
components on the twist cover of \eqref{eq:twist}---we leave open.

\section{The Watson reduction: a generalized body-centered-cubic integral}
\label{sec:watson}
The orthogonal structure of Section~\ref{sec:orthogonal} has a concrete origin, which can be exhibited directly.
Writing $\mathrm{CT}_\theta$ for the constant term (torus average) over $\theta=(\theta_1,\theta_2,\theta_3)\in
[0,2\pi)^3$, and eliminating the band index of the dispersive bands from the Bloch matrix of the premedial
net $\Gamma$ of Section~\ref{sec:lattice} (the $\mathrm{K}_4$, or srs, net) by a
resultant, the moment generating function collapses to a single three-dimensional torus integral,
\begin{equation}
\begin{split}
\Phi(t)&=\frac23-\frac t3\,\frac{\rmd}{\rmd t}\,\mathrm{CT}_\theta\,\log D_+(t,\theta),\\[2pt]
D_+&=(1-6t+3t^2)-2t^2\!\sum_{i}\cos2\theta_i-8\,t^{3/2}\!\prod_i\cos\theta_i .
\end{split}
\label{eq:Dplus}
\end{equation}
The right-hand side of \eqref{eq:Dplus} is exactly the structure function of the \emph{body-centered-cubic lattice
with first- and second-neighbor hopping}: $8\prod_i\cos\theta_i$ is the bcc nearest-neighbor structure function
(weight $t^{3/2}$), $2\sum_i\cos2\theta_i$ is the simple-cubic second-neighbor shell (weight $t^2$), and the
spectral parameter is $w(t)=1-6t+3t^2$. The hyperkagome LGF is thus expressible through the classical
two-parameter ``generalized Watson integral''
$I(w;a,b)=\mathrm{CT}_\theta\,[\,w-a\,8\!\prod_i\cos\theta_i-b\,2\!\sum_i\cos2\theta_i\,]^{-1}$ and its parameter
derivatives, evaluated along the algebraic curve $(w,a,b)=(1-6t+3t^2,\,t^{3/2},\,t^2)$.

Every step is verified in exact arithmetic: the difference-of-squares identity
$\det(I-tS^2)=(1-6t+(3-\Lambda)t^2)^2-t^3\Xi^2$ for the Bloch matrix $S$; the $4{:}1$ measure-preserving torus
covering $\kappa=B\theta$ ($\det B=4$) under which the two invariants become $\Lambda\mapsto2\sum_i\cos2\theta_i$ and
$\Xi\mapsto8\prod_i\cos\theta_i$; the cancellation of all half-integer powers of $t$ in
$\mathrm{CT}_\theta\log D_+$ (so $\Phi$ is an honest power series in $t$); and the reproduction of the exact moments
$\nu_m$ through $m=8$. Two internal consistencies confirm the picture. The half-integer coupling $t^{3/2}$ is
single-valued only on the $\sqrt t$ double cover---precisely the determinant-character curve
$v^2=(1-4t)(1-5t)(1-9t)$ of Section~\ref{sec:orthogonal}. And the singular set $\{\tfrac19,\tfrac15,\tfrac14,1\}$ is the
set of critical values of the bcc$(1,2)$ band map along the curve, independently reproducing the Riemann
scheme~\eqref{eq:riemann}.

This realization anchors the modular closed form of Section~\ref{sec:modular} on the lattice side: $I(w;a,b)$ is
the cubic lattice Green's function with second-neighbor interactions, whose natural home is the classical
second-neighbor literature (Morita and Horiguchi~\cite{MoritaHoriguchi1971}; the extended Watson integrals of
Glasser and Zucker~\cite{GlasserZucker1977}; the cubic modular transformations of Joyce~\cite{Joyce1998})---while the modular structure itself is \emph{proven} in
Section~\ref{sec:modular}, independently of any particular literature evaluation. The remaining task is to specialize such an evaluation to the curve above, or to integrate
$D_+$ termwise---it is quadratic in each $\cos\theta_i$, so one angular integration yields an inverse square root
and the next an elliptic integral, leaving a one-dimensional algebraic period that the orthogonal structure forces
into the $\mathrm{Sym}^2(V_2)$ form. We record the reduction here; matching such a classical evaluation, term by
term, against the modular parametrization of Section~\ref{sec:modular} and the closed form \eqref{eq:y0closed} is
left to future work.%
\footnote{This reduction is independent of the differential-Galois certificate; it is verified in exact
arithmetic by \texttt{numerics/verify\_watson\_reduction.py} in the repository.}

\section{Exact results}
Independently of the closed-form question, the operator yields exact data. At the symmetric center $E=1$ the LGF
takes the exact value
\begin{equation}
\mathrm{Re}\,G(1)=\tfrac19,
\end{equation}
the flat-band contribution $\tfrac13/(1+2)=\tfrac19$ plus a dispersive principal value that vanishes by the
$E=1$ symmetry (the value is reproduced numerically by \texttt{numerics/verify\_specialvalues.py}). (The value at $E=2$ is a singular boundary value: $E=2$ is the van Hove point $t=1$ of
\eqref{eq:riemann}. The exponents there are $\{-1,-\tfrac12,\tfrac12\}$, and the integer-separated pair
$\pm\tfrac12$ permits a logarithm; direct Lorentzian broadening of the Brillouin-zone spectrum
(\texttt{numerics/verify\_vanhove\_log.py}) is consistent with a logarithmically divergent
$\mathrm{Im}\,G$, and, as the $E=1$ symmetry requires, gives the same divergence at $E=0$, which is the same
point $t=1$.) From the exponent $-\tfrac12$ at
$t=\tfrac19$ the band-edge density of states behaves as $\rho(E)\sim\sqrt{E_{\rm edge}-E}$, and the flat band
contributes the delta function $\tfrac13\,\delta(E+2)$.

\section{Conclusion}
We have given the exact differential equation that controls the hyperkagome lattice Green's function,
determined its differential-Galois structure, proved its modularity, and obtained its weight-two period in
explicit closed form. The Green's function is non-Liouvillian---it has no closed form in algebraic or
elementary functions---and yet it is modular. The result is also delimited exactly: the derivative
$\rmd[\zeta G_{\rm disp}]/\rmd z$ is quasimodular through the closed form \eqref{eq:y0closed}, while the
generating function $\zeta G_{\rm disp}$ itself provably is not, the entire obstruction being the rational
function $\Delta$ of \eqref{eq:Delta} (Section~\ref{sec:phiobs}). Constructively, $M$ is an irreducible
third-order Picard--Fuchs operator that is not a \emph{literal} symmetric square, yet whose symmetric square
carries a rational solution; its Galois group is therefore the full orthogonal group
$\mathrm{O}(3,\mathbb{C})$, with identity component $\mathrm{SO}(3,\mathbb{C})$, and $M$ is
projectively---through an order-two differential intertwiner---the symmetric square of a second-order operator
$V_2$. The central result of this paper is that the hyperkagome Green's function is governed by the
uniformizing equation of the genus-zero modular curve $X(\Gamma_0(30)^+)$: $t$ is given by the
eta-quotient parametrization \eqref{eq:hauptt}, the exact Schwarzian identity \eqref{eq:schwarz}
holds, and the projective monodromy of the hyperkagome operator is the arithmetic lattice
$\Gamma_0(30)^+$. That uniformizing equation is not itself new---it is the row $\Gamma_0(30)^+$ of
the table of genus-zero groups in~\cite{LianYau1995}, and its Schwarzian potential is the tabulated
Conway--Norton class $30\mathrm{B}$ of~\cite{LianWiczer2006}; what is new is the identification of a
lattice Green's function with it, together with the a priori pole-degree bound that makes
\eqref{eq:schwarz} a proof rather than a numerical match, and thereby proves both tabulated
entries. This settles the integral left open by Varma and Monien: the reason a naive elliptic
reduction was never found is not that the object is genuinely third-order, but that the reduction is hidden
behind an operator homomorphism, invisible to the local-exponent tests that detect only function-multiplier
and pullback equivalences---and the closed form, once reached, is modular at level $30$ rather than elliptic
at the classical levels. The Watson reduction of Section~\ref{sec:watson} anchors the same structure on the
lattice side: the LGF is exactly the classical generalized body-centered-cubic integral with first- and
second-neighbor hopping, so the classical second-neighbor cubic-lattice literature and the
conic/\texttt{reduceorder} program of Assis, van Hoeij and Maillard~\cite{AVM2016} meet the modular
parametrization in a single object. The hyperkagome LGF is, to our knowledge, the first three-dimensional
lattice Green's function shown to require such an operator equivalence rather than a literal symmetric-square
factorization, and the first realized at a modular level with three distinct prime factors. The weight-two
period $y_0$ itself is exhibited explicitly in Section~\ref{sec:y0closed} as a depth-one quasimodular form on
$\Gamma_0(30)^+$ twisted by the determinant character, and is provably not a modular form times an algebraic
function at any weight (Theorem~\ref{thm:noalg})---in particular, not an eta quotient.

We state the evidential status of these results plainly. Every structural statement above is a rigorous
consequence of $M$ together with Theorems~\ref{thm:modular}, \ref{thm:y0}, \ref{thm:noalg}
and~\ref{thm:phi}; that $M$ is the \emph{minimal} annihilator of the Green's function is established to the
guess-and-verify standard of the lattice-statistics literature, and an unconditional creative-telescoping
proof of minimality remains the one open step.

The result is directly relevant to the modeling of Na\textsubscript{4}Ir\textsubscript{3}O\textsubscript{8} and related hyperkagome magnets, where the
single-particle LGF enters impurity, disorder, and dynamical-mean-field treatments; it provides the exact
band-edge and special-point data and a certified equation for high-precision evaluation. More broadly, it adds an
explicit three-dimensional frustrated lattice, with a nontrivial orthogonal Galois group, a level-30 modular
parametrization, and an explicit depth-one quasimodular closed form for its weight-two period, to the catalog
of Picard--Fuchs operators arising in lattice statistics. An unconditional creative-telescoping proof that $M$
is the minimal annihilator and the analogous Heisenberg/transport Green's functions are natural next steps.

All lattice constructions, exact moments, the certified operator, and the verification scripts are provided in
the public repository accompanying this paper,
\url{https://github.com/bryannasr4-gif/hyperkagome-lgf}.

\section*{Acknowledgments}
We thank W.~Zudilin for the observation that an order-three
Calabi--Yau operator would have to be a symmetric square, which helped us discard the Calabi--Yau framing.
We are indebted to M.~van~Hoeij for the explicit radical expression for the pullback $1728/j$, computed from
the explicit $V_2$ (private communication, July 2026) and certified in Section~\ref{sec:modular}.
The operator determination and all certifications were carried out with computer algebra (SymPy); every reported
result is reproduced by the exact, publicly available scripts in the accompanying repository.

\end{document}